\documentclass[11pt, reco]{article}
\usepackage{amssymb,amscd,amsmath,amsthm}

\usepackage{amsmath, amssymb, amsthm, mathtools}
\usepackage{mathrsfs}
\usepackage{mathrsfs}

\usepackage[colorinlistoftodos]{todonotes}
\usepackage[colorlinks=true, allcolors=blue]{hyperref}
\usepackage{lmodern}
\usepackage[english]{babel}
\usepackage{color}
\usepackage{bm}
\usepackage{graphicx}
\usepackage{tikz}
\usepackage[margin=1in]{geometry}
\newtheorem{theorem}{Theorem}[section]

\newtheorem{corollary}[theorem]{Corollary}
\newtheorem{lemma}[theorem]{Lemma}
\newtheorem{definition}[theorem]{Definition}

\newtheorem{remark}[theorem]{Remark}

 \usepackage{algpseudocode}
 \usepackage{algorithmicx}
 \usepackage{algorithm}

\begin{document}

\title{Causal Architecture in Hidden Quantum Markov Models}

\maketitle
\centerline{ \author{\Large Abdessatar Souissi}}

\centerline{Department of Management Information Systems, College of Business and Economics, }
\centerline{Qassim University, Buraydah 51452, Saudi Arabia}
\centerline{\textit{a.souaissi@qu.edu.sa}}
\vskip0.3cm

\centerline{\author{\Large Abdessatar Barhoumi }}
\centerline{Department of Mathematics and Statistics, College of Science,}
\centerline{  King Faisal University, Al-Ahsa PO.Box: 400, 31982, Saudi Arabia }
\centerline{\textit{abarhoumi@kfu.edu.sa }}
\vskip0.5cm

\begin{abstract}
We introduce a class of causal hidden quantum Markov models (cHQMMs) that reverse the usual order of hidden updates and emissions compared to conventional HQMMs. Using a simple qubit model with a rotating hidden state and sharp measurements, we show that these two architectures—emission then transition versus transition then emission—generally produce different quantum processes. They can be distinguished by measurements at arbitrarily late times, no matter how the hidden system is initialized, and even when the two models start from different initial states. This means that the two orders of operations lead to genuinely different observable behaviors that cannot be reconciled by waiting longer or by choosing special initial conditions. At the same time, we prove that the two architectures become equivalent when they arise from entangled liftings of classical hidden Markov models, sharing the same classical statistics. This identifies a clear dividing line between classical and genuinely quantum hidden memory. Our findings highlight causal HQMMs as a useful tool for studying and distinguishing quantum memory effects in sequential processes.
\end{abstract}

\section{Introduction}
Many time-dependent phenomena are best thought of as being driven by an underlying “memory’’ that keeps track of what happened before and feeds this information forward in time. Hidden Markov models (HMMs) \cite{Rab2002,Mor2021} make this idea precise: a latent chain \((X_n)_{n\in\mathbb{N}}\) plays the role of the memory, and a sequence of observations \((Y_n)_{n\in\mathbb{N}}\) records what we actually measure at each step. At every time \(m\), the update from \(X_{m-1}\) to \((X_m,Y_m)\) factorises as
\[
\mathbb{P}(x_m,y_m\mid x_{m-1})
=
\underbrace{\mathbb{P}(x_m \mid x_{m-1})}_{\text{transition}}
\;\underbrace{\mathbb{P}(y_m \mid x_m)}_{\text{emission}},
\qquad
Y_m\;\perp\!\!\!\perp\;X_{m+1}\,\big\vert\,X_m,
\]
so the hidden chain is exactly the finite-memory mechanism that generates all temporal correlations in the data. Because the transition and emission terms are ordinary conditional probabilities, it makes no practical difference whether we write “transition then emission’’ or “emission then transition’’: the causal content is entirely in these scalar laws, and this commutativity is what keeps classical  HMMs both conceptually simple and computationally efficient \cite{V03,FV05}.

In the quantum setting, the picture becomes substantially more delicate: quantum operations are intrinsically non-commutative, and causal structure has to be formulated directly in terms of completely positive maps, rather than scalar conditional probabilities \cite{WechsEtAl21,WechsEtAl21b}. Several extensions of hidden quantum Markov models (HQMMs) have been developed in this spirit, ranging from many-body and tensor-network descriptions of quantum states to quantum machine-learning models for sequential data and quantum information tasks \cite{Monras11,Srin2017,Clark15}. A fully quantum–probabilistic framework for HQMMs was established in \cite{AGLS24Q}, in which the hidden dynamics is a quantum Markov chain in the sense of \cite{Acc74,ASS20} and the diagonal restriction of the resulting process reproduces a genuine classical HMM. This hybrid structure—quantum memory with a classical “shadow’’—is particularly attractive for near-term implementations, where hybrid quantum–classical architectures and NISQ devices are already being used to simulate noisy channels and to design robust algorithms under realistic noise and resource constraints \cite{David24,De22,Bab22}.

Concretely, the hidden system is described by local matrix algebras \(\mathcal{B}_{H;n}\cong\mathcal{B}(\mathcal{H})\) (with \(\dim\mathcal{H}=N\) and a preferred orthonormal basis \(\{|i\rangle\}_{1\le i\le N}\)), assembled into an infinite tensor product \(\bigotimes_{n\in\mathbb{N}}\mathcal{B}_{H;n}\), and evolves via a family of transition expectations
\[
\mathcal{E}_{H;n}:\mathcal{B}_{H;n}\otimes\mathcal{B}_{H;n+1}\to\mathcal{B}_{H;n},
\]
which are completely positive and unital, and thus admit dual quantum channels \(\mathcal{E}_{H;n}^*:\mathfrak{S}(\mathcal{H})\to\mathfrak{S}(\mathcal{H}\otimes\mathcal{H})\), where \(\mathfrak{S}(\mathcal{H})\) denotes the set of density operators on \(\mathcal{H}\), i.e. all positive semidefinite trace-one operators on \(\mathcal{H}\). The observable system is encoded in local algebras \(\mathcal{B}_{O;n}\cong\mathcal{B}(\mathcal{K})\) (with \(\dim\mathcal{K}=M\) and basis \(\{|e_j\rangle\}_{1\le j\le M}\)), forming \(\bigotimes_{n\in\mathbb{N}}\mathcal{B}_{O;n}\), and couples to the hidden chain through emission expectations
\[
\mathcal{E}_{H,O;n}:\mathcal{B}_{H;n}\otimes\mathcal{B}_{O;n}\to\mathcal{B}_{H;n},
\]
which are again completely positive and unital, with dual channels \(\mathcal{E}_{H,O;n}^*:\mathfrak{S}(\mathcal{H})\to\mathfrak{S}(\mathcal{H}\otimes\mathcal{K})\).

The non-commutativity of the order in which we apply \(\mathcal{E}_{H;n}\) and \(\mathcal{E}_{H,O;n}\) means that the two causal prescriptions—“emission–then–transition’’ and “transition–then–emission’’—are not just two ways of writing the same update, but two genuinely different ways of wiring the same local maps into a one-step HQMM. In \cite{AGLS24Q} the emission–then–transition prescription was analysed in depth, leading to HQMMs with non-trivial entanglement structure \cite{SS23} and to a flexible representation of important classes of matrix product states \cite{Sou25}; related constructions have been shown to capture, in particular, the entanglement and topological features of AKLT-type ground states \cite{SouissiAndolsi25}.

In this paper we compare two causal architectures for HQMMs: emission–then–transition (conventional) and  transition–then–emission (causal). We show that, even with the same local quantum maps $\mathcal{E}_{H;n}$ and $\mathcal{E}_{H,O;n}$, these two time orderings can generate genuinely different dynamics. For the conventional HQMM we take inputs \(X_{H,O}\in \mathcal{B}_{H;n}\otimes\mathcal{B}_{O;n}\) and \(X_{H; n+1}\in\mathcal{B}_{H;n+1}\), and define the emission–then–transition block map
\[
\mathcal{F}^{(n)}(X_{H,O}\otimes X_{n+1})
=
\mathcal{E}_{H;n}\bigl(\mathcal{E}_{H,O;n}(X_{H,O})\otimes X_{H; n+1}\bigr),
\]
while for the causal HQMM we take \(X_{H; n,n+1}\in\mathcal{B}_{H;n}\otimes\mathcal{B}_{H; n+1}\) and \(X_{O}\in\mathcal{B}_{O;n}\), and define the transition–then–emission map
\[
\mathcal{G}^{(n)}(X_{H;n,n+1}\otimes X_{O})
=
\mathcal{E}_{H,O;n}\bigl(\mathcal{E}_{H;n}(X_{H; n,n+1})\otimes X_{O}\bigr).
\]
The adjective “causal’’ is used here in a minimal, operator sense: in \(\mathcal{G}^{(n)}\) the observation at time \(n\) depends on the hidden degrees of freedom only through the updated channel \(\mathcal{E}_{H;n}\), so the output is caused by the post-update hidden configuration, whereas in \(\mathcal{F}^{(n)}\) it is caused by the pre-update one. In a classical HMM, where conditional probabilities commute, this ordering is invisible; in the quantum case the non-commutativity of \(\mathcal{E}_{H;n}\) and \(\mathcal{E}_{H,O;n}\) makes the two wirings genuinely different and, in principle, operationally distinguishable.

To compare the two architectures more concretely, we fix effects in the Heisenberg picture and look at the corresponding dual maps in the Schr\"odinger picture. Let $\mathsf{Eff}(\mathcal{H})$ denote the effect algebra $\{a\in\mathcal{B}(\mathcal{H})\mid 0\le a\le\mathbb{I}\}$ on the hidden Hilbert space and $\mathsf{Eff}(\mathcal{K})$ the set of effects on the output space \cite{MilP07, Kraus83}. For fixed effects $a\in\mathsf{Eff}(\mathcal{H})$ and $b\in\mathsf{Eff}(\mathcal{K})$, inserting $a$ and $b$ in the appropriate legs and then dualising defines one-step hidden channels $\mathcal{F}_{a,b}^{(n)*}$ and $\mathcal{G}_{a,b}^{(n)*}$ acting on hidden states: informally, $\mathcal{F}_{a,b}^{(n)*}$ is the Schr\"odinger dual of the map $X\mapsto\mathcal{F}^{(n)}\bigl((a\otimes b)\otimes X\bigr)$, and $\mathcal{G}_{a,b}^{(n)*}$ is the dual of $X\mapsto\mathcal{G}^{(n)}\bigl((a\otimes X)\otimes b\bigr)$. While these one-step channels already exhibit differences in their Choi–Jamio{\l}kowski operators \cite{Jam72, Choi75}, diamond distances, and entanglement spectra for generic choices of hidden unitaries and non-trivial effects, a much stronger structural separation emerges at the level of the infinite-time joint states. Specifically, we prove in Theorem~\ref{thm:non-quasi-equivalence} that for a minimal qubit model (hidden and output both two-dimensional) \cite{G23, GMM25, YILSDH15} with hidden unitary $U = \exp(-\mathrm{i}\frac{\theta}{2}\sigma_x)$ ($0<|\theta|<\pi$) and sharp measurement in the computational basis, the conventional and causal HQMMs produce joint states $\varphi_{H,O}$ and $\psi_{H,O}$ that are not quasi-equivalent in the sense of Bratteli–Robinson \cite{BR}. This holds for arbitrary initial hidden states $\phi_{H,0}^{\mathrm{conv}}$ and $\phi_{H,0}^{\mathrm{caus}}$ (not necessarily equal). More precisely, there exists a positive constant $\delta > 0$, independent of the observable, such that for every finite time $N_0$ one can find a local observable $A$ supported entirely after $N_0$ with $\|A\| = 1$ and
\begin{equation*}
|\varphi_{H,O}(A) - \psi_{H,O}(A)| = \delta
\end{equation*}
This violates the  asymptotic agreement criterion for quasi-equivalence \cite{BR}, which requires that for every $\varepsilon > 0$ there exists $N_0$ such that all observables supported after $N_0$ have expectation differences less than $\varepsilon$. The result transcends mere differences in one-step channels or finite-time correlation functions; it demonstrates that the two architectures remain distinguishable at arbitrarily late times, irrespective of the choice of initial hidden states, and that no measurement strategy confined to any finite time window can perfectly emulate the other model's future predictions. Thus, the conventional and causal HQMMs define genuinely distinct classes of non-commutative stochastic processes with asymptotically separated joint states in the operator-algebraic sense.

A particularly clear situation where the two causal stories agree at the classical level is provided by entangled hidden Markov models \cite{SS23}. Starting from a classical HMM \(\lambda=(\boldsymbol{\pi},\mathbf{\Pi},\mathbf{Q})\) with hidden state set \(\mathbb{I}_H\), output set \(\mathbb{I}_O\), transition matrices \(\Pi_n=(\Pi_{n;ij})\) and emission kernels \(Q^{(n)}_j(k)\), we realize the hidden system on \(\mathcal{H}\) with basis \(\{|i\rangle\}_{i\in\mathbb{I}_H}\) and outputs on \(\mathcal{K}\) with basis \(\{|e_k\rangle\}_{k\in\mathbb{I}_O}\). The   transitions of the hidden entangled Markov chains \cite{AccFid05, AccMaOhy06} are encoded by a partial isometry \(V_{H;n}:\mathcal{H}\to\mathcal{H}\otimes\mathcal{H}\), \(V_{H;n}|i\rangle=\sum_{j\in\mathbb{I}_H}\sqrt{\Pi_{n;ij}}\,|i,j\rangle\), inducing \(\mathcal{E}_{H;n}(X)=V_{H;n}^{*}X V_{H;n}\); the emissions are encoded by \(V_{H,O;n}:\mathcal{H}\to\mathcal{H}\otimes\mathcal{K}\), \(V_{H,O;n}|j\rangle=\sum_{k\in\mathbb{I}_O}\sqrt{Q^{(n)}_j(k)}\,|j,e_k\rangle\), inducing \(\mathcal{E}_{H,O;n}(Y)=V_{H,O;n}^{*}Y V_{H,O;n}\). These expectations preserve the diagonal matrix units in the preferred bases and reproduce exactly the classical transition and emission probabilities on diagonals. In this entangled HMM setting, one checks on diagonal tensors such as \(|i\rangle\langle i|\otimes|e_k\rangle\langle e_k|\) and their time-shifted versions that \(\mathcal{F}^{(n)}\) and \(\mathcal{G}^{(n)}\) act identically on all diagonal observables, so the associated conventional and causal HQMMs induce the same classical law on the diagonal subalgebra, even though they still define distinct quantum processes on off-diagonal operators.

From a quantum-memory perspective, the way we wire the hidden and emission maps determines how information is stored, propagated, and accessed over time. In topological quantum memories, for example, information is encoded in protected global degrees of freedom and becomes accessible only through carefully structured operations, highlighting how storage and readout are constrained by the underlying dynamics \cite{DenisKit02,FLW24}. Recent work on hidden quantum memory makes a similar point in a different setting: even when the observed statistics appear Markovian, the generating process can still rely on genuinely quantum memory resources that are only revealed by suitable probes \cite{TaElM,Elliot21}.

Viewed in this light, our two causal architectures are not just alternative parameterisations of a model, but two distinct ways of embedding and interrogating finite quantum memory. They provide explicit, operator-level mechanisms for deciding when and how past information is written into the hidden system and when it can be read out through the observations, and thus tie directly into broader efforts to design efficient, robust memory in quantum processes and adaptive agents \cite{Elliott22,SE26,ZBB25}.

The paper is organised as follows. Section \ref{Sec-Pre} reviews  basic tools from   operator algebras and quantum information. Section \ref{Sec-CausalHQMM} introduces causal HQMMs and discuss the two causal architectures. Section \ref{Sec-Qubit} presents a minimal qubit HQMM separating these architectures at the channel level. Section \ref{Sec-EHMM} studies entangled HMMs and identifies a regime where the architectures agree. Section \ref{Sec-disc} discusses implications and future directions.

\section{Preliminaries}\label{Sec-Pre}
Let \(\mathcal{A}\) and \(\mathcal{B}\) be finite-dimensional C\(^*\)-algebras.
A linear map \(\mathcal{E}:\mathcal{A} \to\mathcal{B} \) is called positive if \(\mathcal{E}(X)\ge 0\) whenever \(X\ge 0\). It is called \emph{completely positive} (CP) if, for every \(n\in\mathbb{N}\), its matrix amplification \(\mathcal{E}^{(n)}:M_n(\mathcal{A})\to M_n(\mathcal{B})\), defined by \(\mathcal{E}^{(n)}([X_{ij}]) = [\mathcal{E}(X_{ij})]\), is positive.
 A linear map \(\mathcal{E}:\mathcal{A}\otimes\mathcal{B}\to\mathcal{A}\) is called a \emph{transition expectation} if it is completely positive and normalised, in the sense that \(\mathcal{E}(\mathbb{I}_{\mathcal{A}}\otimes\mathbb{I}_{\mathcal{B}})=\mathbb{I}_{\mathcal{A}}\). This notion, introduced in the operator-algebraic formulation of QMCs  \cite{Acc74,ASS20}. One may regard \(\mathcal{E}\) as a non-commutative analogue of a conditional expectation with respect to the second tensor component: it “integrates out’’ the algebra \(\mathcal{B}\) while preserving the unit and the order structure of \(\mathcal{A}\).

 Throughout, all Hilbert spaces are complex and finite-dimensional. For a Hilbert space \(\mathcal{H}\), we denote by \(\mathcal{B}(\mathcal{H})\) the C\(^*\)-algebra of all bounded linear operators on \(\mathcal{H}\).
 In the concrete case \(\mathcal{A}=\mathcal{B}(\mathcal{H})\) and \(\mathcal{B}=\mathcal{B}(\mathcal{K})\), a transition expectation \(\mathcal{E}:\mathcal{B}(\mathcal{H})\otimes\mathcal{B}(\mathcal{K})\to\mathcal{B}(\mathcal{H})\) can be written, in finite dimensions, in Kraus form as
\[\mathcal{E}(X)=\sum_{i=1}^r V_i^* X V_i,\quad \forall X\in\mathcal{B}(\mathcal{H}\otimes\mathcal{K})
\]
where the Kraus operators \(V_i:\mathcal{H}\to\mathcal{H}\otimes\mathcal{K}\) satisfy \(\sum_i V_i^*V_i=\mathbb{I}_{\mathcal{H}}\),
where the Kraus operators \(V_i:\mathcal{K}\to\mathcal{H}\) satisfy \(\sum_i V_i^*V_i=\mathbb{I}_{\mathcal{H}}\). We distinguish density operators, subnormalised states, and effects using the L\"owner order on self-adjoint operators. For $A,B\in\mathcal{B}(\mathcal{H})_{\mathrm{sa}}$ we write
\[
A \preceq B
\quad\text{iff}\quad
B-A \succeq 0,
\]
that is, $B-A$ is positive semidefinite.

The set of (normalised) density operators on $\mathcal{H}$ is
\[
\mathfrak{S}(\mathcal{H})
:=
\bigl\{\rho\in\mathcal{B}(\mathcal{H})
\;\big|\;
\rho \succeq 0,\ \operatorname{Tr}(\rho)=1\bigr\},
\]
and represents the quantum states of $\mathcal{H}$. The \emph{effect set} on $\mathcal{H}$ is defined by
\[
\mathsf{Eff}(\mathcal{H})
:=
\bigl\{e\in\mathcal{B}(\mathcal{H})
\;\big|\;
0 \preceq e \preceq \mathbb{I}_{\mathcal{H}}\bigr\},
\]
and its elements represent quantum events, i.e.\ yes–no measurements with outcome probabilities $\operatorname{Tr}(\rho e)$ for $\rho\in\mathfrak{S}(\mathcal{H})$.
The Hilbert–Schmidt inner product \(\langle\rho,X\rangle=\operatorname{Tr}(\rho X)\) on \(\mathfrak{S}(\mathcal{H})\times\mathcal{B}(\mathcal{H})\) induces a unique adjoint map \(\mathcal{E}_*:\mathfrak{S}(\mathcal{H})\to\mathfrak{S}(\mathcal{H}\otimes\mathcal{K})\) defined by the single duality relation
\begin{equation}\label{eq_HS}
\operatorname{Tr}(\mathcal{E}_*(\rho)X)=\operatorname{Tr}(\rho\,\mathcal{E}(X))
\end{equation}
for all \(\rho\in\mathfrak{S}(\mathcal{H})\) and all \(X\in\mathcal{B}(\mathcal{H}\otimes\mathcal{K})\). In Kraus form this adjoint is \(\mathcal{E}_*(\rho)=\sum_i V_i \rho V_i^*\), which is a completely positive, trace-preserving map on states.

At the level of quantum information, a channel \(\Phi:\mathfrak{S}(\mathcal{H})\to\mathfrak{S}(\mathcal{H}')\) represents a causal influence from the input system \(\mathcal{H}\) to the output system \(\mathfrak{S}(\mathcal{H}')\): interventions on the input density matrix can change the statistics of all future measurements on the enlarged system, but not conversely. The dual map \(\mathcal{E}\) propagates observables backwards along this causal arrow, implementing the Heisenberg-picture dynamics of a quantum Markov chain.

We now specialise to the setting of hidden quantum Markov models. Let \(\mathcal{H}\) be an \(N\)-dimensional Hilbert space representing the hidden, or internal, degrees of freedom. We fix an orthonormal basis \(\{|j\rangle : j\in\mathbb{I}_H\}\) with \(\mathbb{I}_H=\{1,\dots,N\}\), and denote the associated hidden observable algebra by
\(
\mathcal{B}_H := \mathcal{B}(\mathcal{H}).
\)
Similarly, let \(\mathcal{K}\) be an \(M\)-dimensional Hilbert space modelling the output or observation space, with orthonormal basis \(\{|e_k\rangle : k\in\mathbb{I}_O\}\), \(\mathbb{I}_O=\{1,\dots,M\}\), and observable algebra
\(
\mathcal{B}_O := \mathcal{B}(\mathcal{K}).
\)
Heuristically, \(\mathcal{B}_H\) contains the “latent quantum causes” and \(\mathcal{B}_O\) contains the “observable effects” in the sense of quantum cause–effect identification.

Time is discrete and indexed by \(\mathbb{N}=\{0,1,2,\dots\}\). At each time \(n\in\mathbb{N}\), we attach Hilbert spaces \(\mathcal{H}_n\simeq\mathcal{H}\) and \(\mathcal{K}_n\simeq\mathcal{K}\) with observable algebras
\[
\mathcal{B}_{H;n} := \mathcal{B}(\mathcal{H}_n)\simeq\mathcal{B}_H,\qquad
\mathcal{B}_{O;n} := \mathcal{B}(\mathcal{K}_n)\simeq\mathcal{B}_O.
\]
The local hidden–output algebra at time \(n\) is the spatial tensor product
\(\mathcal{B}_{H,O;n} := \mathcal{B}_{H;n}\otimes\mathcal{B}_{O;n},\)
which encodes all observables acting jointly on the hidden subsystem and its associated output at that instant.
For a finite time interval \([0,n]=\{0,1,\dots,n\}\), we consider the tensor product algebras
\[
\mathcal{B}_{H;[0,n]} := \bigotimes_{m=0}^n \mathcal{B}_{H;m},\qquad
\mathcal{B}_{O;[0,n]} := \bigotimes_{m=0}^n \mathcal{B}_{O;m},\qquad
\mathcal{B}_{H,O;[0,n]} := \bigotimes_{m=0}^n \mathcal{B}_{H,O;m},
\]
describing, respectively, the hidden history, the observable history, and the full hidden–observable history over \([0,n]\). The corresponding local (finitely supported) algebras are the algebraic inductive limits
\[
\mathcal{B}_{H;\mathrm{loc}} := \bigcup_{n\ge 0} \mathcal{B}_{H;[0,n]},\quad
\mathcal{B}_{O;\mathrm{loc}} := \bigcup_{n\ge 0} \mathcal{B}_{O;[0,n]},\quad
\mathcal{B}_{H,O;\mathrm{loc}} := \bigcup_{n\ge 0} \mathcal{B}_{H,O;[0,n]}.
\]
Their norm closures in the operator norm yield the quasi-local C\(^*\)-algebras
\[
\mathcal{B}_{H;\mathbb{N}} :=  \bigotimes_{n\in\mathbb{N}}\mathcal{B}_{H;n},\quad
\mathcal{B}_{O;\mathbb{N}} := \bigotimes_{n\in\mathbb{N}}\mathcal{B}_{O;n},\quad
\mathcal{B}_{H,O;\mathbb{N}} := \bigotimes_{n\in\mathbb{N}}\mathcal{B}_{H,O;n}.
\]
\begin{definition}\label{def:quasi-equivalence}
Let $\mathcal{A} = \overline{\bigcup_{\Lambda \Subset L} \mathcal{A}_\Lambda}^{\|\cdot\|}$ be a quasi-local algebra over a countable index set $L$, where $(\mathcal{A}_\Lambda)_{\Lambda \Subset L}$ is an inductive family of finite-dimensional C$^*$-algebras with natural embeddings for $\Lambda \subset \Lambda'$.
Two states $\varphi_1$ and $\varphi_2$ on $\mathcal{A}$ are called \emph{quasi-equivalent} if for every $\varepsilon > 0$ there exists a finite region $\Lambda_0 \Subset L$ such that for all local observables $a \in \mathcal{A}_\Lambda$ with $\Lambda \cap \Lambda_0 = \varnothing$, the following inequality holds:
\begin{equation}\label{eq:quasi-equiv-condition}
|\varphi_1(a) - \varphi_2(a)| < \varepsilon \|a\|.
\end{equation}
\end{definition}

\begin{remark} 
Conversely, two states $\varphi_1$ and $\varphi_2$ are \emph{not quasi-equivalent} if there exists some $\varepsilon_0 > 0$ such that for every finite region $\Lambda_0 \Subset L$, one can find a local observable $a \in \mathcal{A}_\Lambda$ with $\Lambda \cap \Lambda_0 = \varnothing$ (i.e., $a$ supported arbitrarily far from $\Lambda_0$) satisfying
\[
|\varphi_1(a) - \varphi_2(a)| \ge \varepsilon_0 \|a\|.
\]
In other words, the two states remain distinguishable by measurements performed arbitrarily far away from any bounded region, no matter how large that region is. This typically occurs when $\varphi_1$ and $\varphi_2$ are disjoint (Inequivalent)  states, such as pure phases in a thermodynamic system that are not unitarily equivalent.
\end{remark}

The reader is referred to \cite{BR} for a detailed presentation of quasi-local algebras.
 Within this framework, a hidden quantum Markov model is specified by an initial hidden state \(\phi_{H,0}:\mathcal{B}_H\to\mathbb{C}\) and two families of transition expectations that describe, respectively, the internal propagation of the hidden chain and its causal influence on the output. The hidden dynamics is given by CPIP maps \(\mathcal{E}_{H;n}:\mathcal{B}_{H;n}\otimes\mathcal{B}_{H;n+1}\to\mathcal{B}_{H;n}\) for \(n\ge 0\). Via Hilbert–Schmidt duality, each \(\mathcal{E}_{H;n}\) induces a channel \((\mathcal{E}_{H;n})_*:\mathfrak{S}(\mathcal{H}_n)\to\mathfrak{S}(\mathcal{H}_n\otimes\mathcal{H}_{n+1})\) defined by \(\operatorname{Tr}((\mathcal{E}_{H;n})_*(\rho_{H;n})X)=\operatorname{Tr}(\rho_{H;n}\,\mathcal{E}_{H;n}(X))\) for all \(\rho_{H;n}\in\mathfrak{S}(\mathcal{H}_n)\), \(X\in\mathcal{B}(\mathcal{H}_n\otimes\mathcal{H}_{n+1})\).

In this subsection we recast the comparison between the conventional and causal architectures at a fixed time-step in the language of quantum channel discrimination. We begin with a rigorous definition of the diamond distance between two completely positive trace-preserving (CPTP) maps and then specialize to the dual one-step hidden channels induced by the block maps of the two architectures. This allows us to quantify, in a fully operational way, how distinguishable the two temporal orders are, in the sense of optimal one-shot channel discrimination \cite{PLLP19}.

\begin{definition}[Diamond distance]
Let $\mathcal{H}_{\mathrm{in}}$ and $\mathcal{H}_{\mathrm{out}}$ be finite-dimensional Hilbert spaces and let
\[
\Phi_1,\Phi_2:\mathcal{B}(\mathcal{H}_{\mathrm{in}})\to\mathcal{B}(\mathcal{H}_{\mathrm{out}})
\]
be quantum channels (normal, completely positive and trace-preserving). The \emph{diamond distance} between $\Phi_1$ and $\Phi_2$ is defined by
\[
\|\Phi_1-\Phi_2\|_{\diamond}
:=
\sup_{d\in\mathbb{N}}
\;\sup_{\rho\in\mathfrak{S}(\mathcal{H}_{\mathrm{in}}\otimes\mathbb{C}^d)}
\big\|
(\Phi_1\otimes\mathrm{id}_d)(\rho)
-
(\Phi_2\otimes\mathrm{id}_d)(\rho)
\big\|_1
\]
where $\|\cdot\|_1$ denotes the trace norm and $\mathrm{id}_d$ is the identity channel on $\mathcal{B}(\mathbb{C}^d)$. For finite-dimensional $\mathcal{H}_{\mathrm{in}}$ the supremum may be restricted to $d=\dim\mathcal{H}_{\mathrm{in}}$.
\end{definition}

The diamond distance is the canonical measure of single-use distinguishability of quantum channels. If one is given a single use of an unknown channel, promised to be either $\Phi_1$ or $\Phi_2$ with equal a priori probabilities, then the optimal success probability $p_{\mathrm{succ}}^{(1)}$ in identifying the channel satisfies
\[
p_{\mathrm{succ}}^{(1)}
=
\frac12+\frac14\,\|\Phi_1-\Phi_2\|_{\diamond},
\]
so that $\|\Phi_1-\Phi_2\|_{\diamond}=0$ if and only if $\Phi_1$ and $\Phi_2$ are operationally indistinguishable in one use.

\medskip

We next introduce the Choi--Jamio{\l}kowski representation of a channel and recall its relation to the diamond norm. Fix an orthonormal basis $\{|i\rangle\}_{i=1}^{d_{\mathrm{in}}}$ of $\mathcal{H}_{\mathrm{in}}$ and consider the (unnormalised) maximally entangled vector
\[
|\Omega\rangle
:=
\sum_{i=1}^{d_{\mathrm{in}}}
|i\rangle\otimes|i\rangle
\;\in\;
\mathcal{H}_{\mathrm{in}}\otimes\mathcal{H}_{\mathrm{in}}.
\]

\begin{definition}[Choi operator]
For a linear map
\(
\Phi:\mathcal{B}(\mathcal{H}_{\mathrm{in}})\to\mathcal{B}(\mathcal{H}_{\mathrm{out}})
\),
its Choi operator is
\[
J(\Phi)
:=
(\mathrm{id}\otimes\Phi)\bigl(|\Omega\rangle\langle\Omega|\bigr)
\;\in\;
\mathcal{B}(\mathcal{H}_{\mathrm{in}}\otimes\mathcal{H}_{\mathrm{out}}).
\]
If $\Phi$ is CPTP, then $J(\Phi)\geq 0$ and
\(
\operatorname{Tr}_{\mathrm{out}} J(\Phi)=\mathbb{I}_{\mathcal{H}_{\mathrm{in}}}
\).

\end{definition}

The assignment $\Phi\mapsto J(\Phi)$ is injective and affine, and there are well-known bounds relating the diamond norm of $\Phi_1-\Phi_2$ to the trace norm of the Choi difference $J(\Phi_1)-J(\Phi_2)$. In particular, one has
\begin{equation}\label{eq:diamond-choi-bounds}
\frac{1}{d_{\mathrm{in}}}\,
\bigl\|J(\Phi_1)-J(\Phi_2)\bigr\|_1
\;\le\;
\|\Phi_1-\Phi_2\|_{\diamond}
\;\le\;
\bigl\|J(\Phi_1)-J(\Phi_2)\bigr\|_1
\end{equation}
where $d_{\mathrm{in}}=\dim\mathcal{H}_{\mathrm{in}}$. The lower bound may be viewed as a consequence of evaluating the supremum in the diamond norm on suitably chosen input states that are close to maximally entangled, while the upper bound follows from an operator-space duality argument.

\section{Causal hidden quantum Markov models}\label{Sec-CausalHQMM}
The coupling between hidden and observable degrees of freedom is encoded by emission expectations \(\mathcal{E}_{H,O;n}:\mathcal{B}_{H;n}\otimes\mathcal{B}_{O;n}\to\mathcal{B}_{H;n}\). For \(b_n\in\mathcal{B}_{O;n}\), the dual map \((\mathcal{E}_{H,O;n})_*:\mathfrak{S}(\mathcal{H}_n)\to\mathfrak{S}(\mathcal{H}_n\otimes\mathcal{K}_n)\) is the quantum operation that produces an output system \(\mathcal{K}_n\) conditioned on the hidden state at time \(n\), and formalises the causal arrow \(H_n\to O_n\). Tracing out the hidden algebra yields an induced process on \(\mathcal{B}_{O;\mathbb{N}}\) whose temporal correlations may violate classical Markov conditions and appear “non-Markovian’’; in the process-tensor viewpoint \cite{MilzStrasberg21} this non-Markovianity can be understood as the imprint of genuine quantum memory living in \(\mathcal{B}_{H;\mathbb{N}}\) and propagating via the \(\mathcal{E}_{H;n}\).

Given local elements \(a_n\in\mathcal{B}_{H;n}\), \(a_{n+1}\in\mathcal{B}_{H;n+1}\), and \(b_n\in\mathcal{B}_{O;n}\), the one-step map
\begin{equation}\label{eq_F}
\mathcal{F}_{a_n,b_n}^{(n)}(a_{n+1})
:=
\mathcal{E}_{H;n}\bigl(\mathcal{E}_{H,O;n}(a_n\otimes b_n)\otimes a_{n+1}\bigr),
\end{equation}
realises the conventional HQMM causal order: the hidden system at time \(n\) first influences the output at time \(n\) through \(\mathcal{E}_{H,O;n}\), then is propagated internally to time \(n+1\) via \(\mathcal{E}_{H;n}\). The alternative one-step map
\begin{equation}\label{eq_G}
\mathcal{G}_{a_n,b_n}^{(n)}(a_{n+1})
:=
{\mathcal{E}}_{H,O;n}\bigl(\mathcal{E}_{H;n}(a_n\otimes a_{n+1})\otimes b_n\bigr),
\end{equation}
 built from a second family \(\widetilde{\mathcal{E}}_{H,O;n}\), swaps this order: the hidden chain first advances from \(H_n\) to \(H_{n+1}\) and only the updated hidden configuration acts as the cause for the observation at time \(n\). In both cases, composing the maps \(\mathcal{F}_{a_n,b_n}^{(n)}\) or \(\mathcal{G}_{a_n,b_n}^{(n)}\) and evaluating with \(\phi_{H,0}\) produces a state on \(\mathcal{B}_{H,O;\mathbb{N}}\), and the restriction to \(\mathcal{B}_{O;\mathbb{N}}\) is an output process whose temporal structure reflects the chosen causal order.

 The next pair of diagrams summarises these two one-step updates in a way that makes this causal structure manifest. For clarity and symmetry, the conventional diagram is placed in the upper-left, while the causal diagram is placed in the lower-right; the red and blue blocks are chosen so that the "role" of each block is exchanged between the two diagrams.

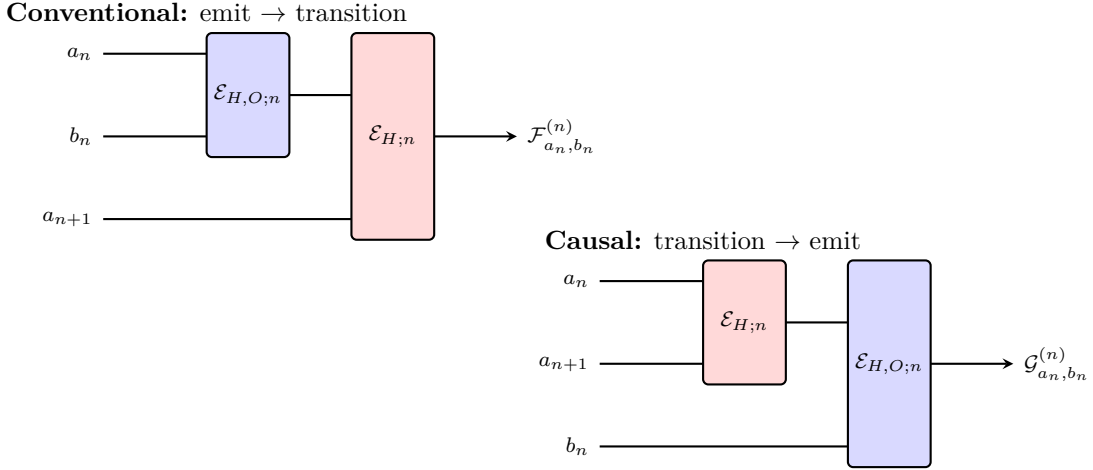
\begin{center}
\begin{tikzpicture}[thick, scale=0.55, >=stealth]

    \node[font=\small] at (2.5,3.5) {\textbf{Conventional:} emit $\rightarrow$ transition};

    \draw (0,2.5) -- (2.5,2.5);
    \draw (0,0.5) -- (2.5,0.5);
    \node[left, font=\scriptsize] at (0,2.5) {$a_n$};
    \node[left, font=\scriptsize] at (0,0.5) {$b_n$};

    \draw[fill=blue!15, rounded corners=2pt] (2.5,0) rectangle (4.5,3);
    \node[font=\scriptsize] at (3.5,1.5) {$\mathcal{E}_{H,O;n}$};

    \draw (4.5,1.5) -- (6,1.5);

    \draw (0,-1.5) -- (6,-1.5);
    \node[left, font=\scriptsize] at (0,-1.5) {$a_{n+1}$};

    \draw[fill=red!15, rounded corners=2pt] (6,-2) rectangle (8,3);
    \node[font=\scriptsize] at (7,0.5) {$\mathcal{E}_{H;n}$};

    \draw[->] (8,0.5) -- (10,0.5);
    \node[right, font=\scriptsize] at (10,0.5) {$\mathcal{F}_{a_n,b_n}^{(n)}$};

    \node[font=\small] at (14.5,-2.0) {\textbf{Causal:} transition $\rightarrow$ emit};

    \draw (12,-3) -- (14.5,-3);
    \draw (12,-5) -- (14.5,-5);
    \node[left, font=\scriptsize] at (12,-3) {$a_n$};
    \node[left, font=\scriptsize] at (12,-5) {$a_{n+1}$};

    \draw[fill=red!15, rounded corners=2pt] (14.5,-5.5) rectangle (16.5,-2.5);
    \node[font=\scriptsize] at (15.5,-4) {$\mathcal{E}_{H;n}$};

    \draw (16.5,-4) -- (18,-4);

    \draw (12,-7) -- (18,-7);
    \node[left, font=\scriptsize] at (12,-7) {$b_n$};

    \draw[fill=blue!15, rounded corners=2pt] (18,-7.5) rectangle (20,-2.5);
    \node[font=\scriptsize] at (19,-5) {${\mathcal{E}}_{H,O;n}$};

    \draw[->] (20,-5) -- (22,-5);
    \node[right, font=\scriptsize] at (22,-5) {$\mathcal{G}_{a_n,b_n}^{(n)}$};

\end{tikzpicture}
\end{center}

Although the two architectures are built from the same families of completely positive maps, they implement distinct causal relationships between hidden and observable subsystems. In particular, since \(\mathcal{E}_{H;n}\) and \(\mathcal{E}_{H,O;n}\)  do not commute in general, the one-step maps \(\mathcal{F}_{a_n,b_n}^{(n)}\) and \(\mathcal{G}_{a_n,b_n}^{(n)}\) give rise to non-equivalent hidden quantum Markov processes.

\begin{definition}\cite{AGLS24Q}[Conventional hidden quantum Markov model]
A \emph{conventional hidden quantum Markov model} (conventional HQMM) is a quadruple
\[
\Xi_{\mathrm{conv}}
=
\bigl(\phi_{H,0},(\mathcal{E}_{H;n})_{n\ge 0},(\mathcal{E}_{H,O;n})_{n\ge 0},(\mathcal{F}^{(n)})_{n\ge 0}\bigr),
\]
where:
\begin{itemize}
\item \(\phi_{H,0}:\mathcal{B}_H\to\mathbb{C}\) is an initial state on the hidden algebra;
\item \(\mathcal{E}_{H;n}:\mathcal{B}_{H;n}\otimes\mathcal{B}_{H;n+1}\to\mathcal{B}_{H;n}\) are hidden transition expectations;
\item \(\mathcal{E}_{H,O;n}:\mathcal{B}_{H;n}\otimes\mathcal{B}_{O;n}\to\mathcal{B}_{H;n}\) are emission expectations;
\item for each time \(n\), the one-step block map
\[
\mathcal{F}^{(n)}:\mathcal{B}_{H;n}\otimes\mathcal{B}_{O;n}\otimes\mathcal{B}_{H;n+1}\to\mathcal{B}_{H;n}
\]
is defined, for local inputs \(a_n\in\mathcal{B}_{H;n}\), \(b_n\in\mathcal{B}_{O;n}\), \(X\in\mathcal{B}_{H;n+1}\), by
\[
\mathcal{F}^{(n)}_{a_n,b_n}(X)
:=
\mathcal{E}_{H;n}\bigl(\mathcal{E}_{H,O;n}(a_n\otimes b_n)\otimes X\bigr).
\]
\end{itemize}
For a local tensor observable \(\bigotimes_{m=0}^n (a_m\otimes b_m)\in\mathcal{B}_{H,O;[0,n]}\), the finite-time joint expectation of the conventional HQMM is given by
\begin{equation}\label{eq:conv-state-no-sup}
\varphi_{H,O}\!\left(\bigotimes_{m=0}^n (a_m\otimes b_m)\right)
:=
\phi_{H,0}\circ
\mathcal{F}^{(0)}_{a_0,b_0}\circ\mathcal{F}^{(1)}_{a_1,b_1}\circ\cdots\circ
\mathcal{F}^{(n)}_{a_n,b_n}(\mathbb{I}_{H;n+1})
\end{equation}
These linear functionals are compatible for different \(n\) and extend uniquely (by the usual projective-limit construction) to a state \(\varphi_{H,O}\) on the quasi-local algebra \(\mathcal{B}_{H,O;\mathbb{N}}\). The corresponding hidden and observable marginals are defined by restriction:
\[
\varphi_H
:=
\varphi_{H,O}\!\restriction_{\mathcal{B}_{H;\mathbb{N}}},\qquad
\varphi_O
:=
\varphi_{H,O}\!\restriction_{\mathcal{B}_{O;\mathbb{N}}}
\]
\end{definition}

\begin{definition}[Causal hidden quantum Markov model]
A \emph{causal hidden quantum Markov model} (causal HQMM) is a quadruple
\[
\Xi_{\mathrm{caus}}
=
\bigl(\phi_{H,0},(\mathcal{E}_{H;n})_{n\ge 0},(\mathcal{E}_{H,O;n})_{n\ge 0},(\mathcal{G}^{(n)})_{n\ge 0}\bigr),
\]
with the same initial hidden state \(\phi_{H,0}\) and the same families of hidden transition and emission expectations \((\mathcal{E}_{H;n})_{n\ge 0}\), \((\mathcal{E}_{H,O;n})_{n\ge 0}\) as above, but a different one-step composition rule encoded in the block maps \(\mathcal{G}^{(n)}\). For each \(n\), the block map
\[
\mathcal{G}^{(n)}:\mathcal{B}_{H;n}\otimes\mathcal{B}_{O;n}\otimes\mathcal{B}_{H;n+1}\to\mathcal{B}_{H;n}
\]
is defined, for \(a_n\in\mathcal{B}_{H;n}\), \(b_n\in\mathcal{B}_{O;n}\), \(X\in\mathcal{B}_{H;n+1}\), by
\[
\mathcal{G}^{(n)}_{a_n,b_n}(X)
:=
\mathcal{E}_{H,O;n}\bigl(\mathcal{E}_{H;n}(a_n\otimes X)\otimes b_n\bigr).
\]
For a local tensor \(\bigotimes_{m=0}^n (a_m\otimes b_m)\in\mathcal{B}_{H,O;[0,n]}\), the finite-time joint expectation of the causal HQMM is then
\begin{equation}\label{eq:caus-state-no-sup}
\psi_{H,O}\!\left(\bigotimes_{m=0}^n (a_m\otimes b_m)\right)
:=
\phi_{H,0}\circ
\mathcal{G}^{(0)}_{a_0,b_0}\circ\mathcal{G}^{(1)}_{a_1,b_1}\circ\cdots\circ
\mathcal{G}^{(n)}_{a_n,b_n}(\mathbb{I}_{H;n+1}),
\end{equation}
which again extends uniquely to a state \(\psi_{H,O}\) on \(\mathcal{B}_{H,O;\mathbb{N}}\). The associated hidden and observable marginals are
\[
\psi_H
:=
\psi_{H,O}\!\restriction_{\mathcal{B}_{H;\mathbb{N}}},\qquad
\psi_O
:=
\psi_{H,O}\!\restriction_{\mathcal{B}_{O;\mathbb{N}}}.
\]
\end{definition}

From the perspective of quantum causal models, the “conventional’’ and “causal’’ HQMM architectures correspond to two inequivalent ways of realising the same underlying static causal graph at the level of channels: in the conventional case the edge \(H_n\to O_n\) is implemented before \(H_n\to H_{n+1}\), whereas in the causal case \(H_n\to H_{n+1}\) is implemented first. Since \(\mathcal{E}_{H;n}\) and \(\mathcal{E}_{H,O;n}\) need not commute, the two resulting processes are generally different. Yet in both architectures, apparent non-Markovianity in the output algebra \(\mathcal{B}_{O;\mathbb{N}}\) can be interpreted as the projection of a Markovian quantum causal dynamics on \(\mathcal{B}_{H,O;\mathbb{N}}\) onto the observable sector, with the hidden algebra \(\mathcal{B}_{H;\mathbb{N}}\) carrying the quantum memory that mediates the causal influence from past hidden configurations to future observations.  This is the HQMM analogue of the no-retrocausality and faithfulness assumptions in quantum causal discovery \cite{GiarmatziCosta18}.

In both cases, the triplet \((\phi_{H,0},(\mathcal{E}_{H;n}),(\mathcal{E}_{H,O;n}))\) fixes the local building blocks (initial state, hidden transitions, and emission expectations), while the families \((\mathcal{F}^{(n)})_{n\ge 0}\) and \((\mathcal{G}^{(n)})_{n\ge 0}\) specify how these blocks are composed in time. It is precisely this choice of block maps—emission–then–transition in the conventional case, transition–then–emission in the causal case—that distinguishes the two HQMM architectures and leads to different joint states \(\varphi_{H,O}\), \(\psi_{H,O}\) and different marginals on the hidden and observable algebras.

\begin{lemma}\label{lem:dual-block-maps} 
Fix a time step $n$ and suppose that the hidden transition and emission expectations admit minimal Kraus decompositions
\[
\mathcal{E}_{H;n}(X)
=
\sum_{\alpha=1}^{r_H} K_{H;\alpha}^{*} X K_{H;\alpha},
\qquad
\mathcal{E}_{H,O;n}(Y)
=
\sum_{\beta=1}^{r_O} K_{H,O;\beta}^{*} Y K_{H,O;\beta},
\]
where $K_{H;\alpha}:\mathcal{H}_n\to\mathcal{H}_n\otimes\mathcal{H}_{n+1}$ and $K_{H,O;\beta}:\mathcal{H}_n\to\mathcal{H}_n\otimes\mathcal{K}_n$ satisfy
\(
\sum_{\alpha}K_{H;\alpha}^{*}K_{H;\alpha}
=
\sum_{\beta}K_{H,O;\beta}^{*}K_{H,O;\beta}
=
\mathbb{I}_{\mathcal{H}_n}
\).
Let $\mathcal{F}_{a,b}^{(n)}$ and $\mathcal{G}_{a,b}^{(n)}$ be the block maps defined in \eqref{eq_F} and \eqref{eq_G}, and let
\[
\mathcal{F}_{a,b}^{(n)*},\mathcal{G}_{a,b}^{(n)*}:\mathfrak{S}_{\le 1}(\mathcal{H}_n)\to\mathfrak{S}(\mathcal{H}_{n+1})
\]
denote their Schr\"odinger-picture duals with respect to the Hilbert–Schmidt pairing \eqref{eq_HS}. Then, for every $\rho\in\mathfrak{S}(\mathcal{H}_n)$,
\begin{align} \label{Fab*}
\mathcal{F}_{a,b}^{(n)*}(\rho)
&=
\sum_{\alpha=1}^{r_H}
\operatorname{Tr}_{\mathcal{H}_n}\Bigl[
K_{H;\alpha}\,\rho\,K_{H;\alpha}^{*}
\bigl(\mathcal{E}_{H,O;n}(a\otimes b)\otimes\mathbb{I}_{\mathcal{H}_{n+1}}\bigr)
\Bigr]\nonumber\\[0.4em]
&=
\sum_{\alpha=1}^{r_H}\sum_{\beta=1}^{r_O}
\operatorname{Tr}_{\mathcal{H}_n}\Bigl[
K_{H;\alpha}\,\rho\,K_{H;\alpha}^{*}
\Bigl(
 K_{H,O;\beta}^{*}(a\otimes b)K_{H,O;\beta}
\otimes\mathbb{I}_{\mathcal{H}_{n+1}}
\Bigr)
\Bigr]\\[0.4em]
\label{Gab*}\mathcal{G}_{a,b}^{(n)*}(\rho)
&=
\sum_{\beta=1}^{r_O}
\operatorname{Tr}_{\mathcal{H}_n}\Bigl[
K_{H,O;\beta}\,\rho\,K_{H,O;\beta}^{*}
\bigl(\mathcal{E}_{H;n}(a\otimes\mathbb{I}_{\mathcal{H}_{n+1}})\otimes b\bigr)
\Bigr]\nonumber\\
&=
\sum_{\alpha=1}^{r_H}\sum_{\beta=1}^{r_O}
\operatorname{Tr}_{\mathcal{H}_n}\Bigl[
K_{H,O;\beta}\,\rho\,K_{H,O;\beta}^{*}
\Bigl(
K_{H;\alpha}^{*}(a\otimes\mathbb{I}_{\mathcal{H}_{n+1}})K_{H;\alpha}
\otimes b
\Bigr)
\Bigr]
\end{align}
where $\operatorname{Tr}_{\mathcal{H}_n}$ denotes the partial trace over the hidden space at time $n$.
\end{lemma}

\begin{proof}
We prove the formula for $\mathcal{F}_{a,b}^{(n)*}$; the argument for $\mathcal{G}_{a,b}^{(n)*}$ is analogous. The duality relation \eqref{eq_HS} states that for all $\rho\in\mathfrak{S}(\mathcal{H}_n)$ and $X\in\mathcal{B}(\mathcal{H}_{n+1})$,
\[
\operatorname{Tr}\bigl(\mathcal{F}_{a,b}^{(n)*}(\rho)\,X\bigr)
=
\operatorname{Tr}\bigl(\rho\,\mathcal{F}_{a,b}^{(n)}(X)\bigr).
\]
Using \eqref{eq_F} and the Kraus representation of $\mathcal{E}_{H;n}$, we obtain
\[
\mathcal{F}_{a,b}^{(n)}(X)
=
\mathcal{E}_{H;n}\bigl(\mathcal{E}_{H,O;n}(a\otimes b)\otimes X\bigr)
=
\sum_{\alpha=1}^{r_H}
K_{H;\alpha}^{*}\bigl(\mathcal{E}_{H,O;n}(a\otimes b)\otimes X\bigr)K_{H;\alpha}.
\]
Hence
\[
\operatorname{Tr}\bigl(\rho\,\mathcal{F}_{a,b}^{(n)}(X)\bigr)
=
\sum_{\alpha=1}^{r_H}
\operatorname{Tr}\Bigl(\rho\,K_{H;\alpha}^{*}\bigl(\mathcal{E}_{H,O;n}(a\otimes b)\otimes X\bigr)K_{H;\alpha}\Bigr).
\]
By cyclicity of the trace, each summand can be rewritten as
\[
\operatorname{Tr}\Bigl(K_{H;\alpha}\rho K_{H;\alpha}^{*}\bigl(\mathcal{E}_{H,O;n}(a\otimes b)\otimes X\bigr)\Bigr).
\]
We now interpret this as a trace over $\mathcal{H}_n\otimes\mathcal{H}_{n+1}$ and factor it via the partial trace. For any operator $M$ on $\mathcal{H}_n\otimes\mathcal{H}_{n+1}$ and $X$ on $\mathcal{H}_{n+1}$, one has
\[
\operatorname{Tr}_{\mathcal{H}_n\otimes\mathcal{H}_{n+1}}\bigl[M(\mathbb{I}_{\mathcal{H}_n}\otimes X)\bigr]
=
\operatorname{Tr}_{\mathcal{H}_{n+1}}\Bigl[\operatorname{Tr}_{\mathcal{H}_n}(M)\,X\Bigr],
\]
so with
\[
M
=
K_{H;\alpha}\rho K_{H;\alpha}^{*}
\bigl(\mathcal{E}_{H,O;n}(a\otimes b)\otimes\mathbb{I}_{\mathcal{H}_{n+1}}\bigr)
\]
we obtain
\[
\operatorname{Tr}\Bigl(K_{H;\alpha}\rho K_{H;\alpha}^{*}\bigl(\mathcal{E}_{H,O;n}(a\otimes b)\otimes X\bigr)\Bigr)
=
\operatorname{Tr}_{\mathcal{H}_{n+1}}\Bigl[
\operatorname{Tr}_{\mathcal{H}_n}\Bigl(
K_{H;\alpha}\rho K_{H;\alpha}^{*}
\bigl(\mathcal{E}_{H,O;n}(a\otimes b)\otimes\mathbb{I}_{\mathcal{H}_{n+1}}\bigr)
\Bigr)X
\Bigr].
\]
Summing over $\alpha$ and using linearity of the partial trace, we arrive at
\[
\operatorname{Tr}\bigl(\rho\,\mathcal{F}_{a,b}^{(n)}(X)\bigr)
=
\operatorname{Tr}_{\mathcal{H}_{n+1}}\Bigl(
\Bigl[
\sum_{\alpha=1}^{r_H}
\operatorname{Tr}_{\mathcal{H}_n}\Bigl(
K_{H;\alpha}\rho K_{H;\alpha}^{*}
\bigl(\mathcal{E}_{H,O;n}(a\otimes b)\otimes\mathbb{I}_{\mathcal{H}_{n+1}}\bigr)
\Bigr)
\Bigr]X
\Bigr).
\]
By the defining duality, the square bracket must coincide with $\mathcal{F}_{a,b}^{(n)*}(\rho)$, since this identity holds for all $X$. Thus
\[
\mathcal{F}_{a,b}^{(n)*}(\rho)
=
\sum_{\alpha=1}^{r_H}
\operatorname{Tr}_{\mathcal{H}_n}\Bigl[
K_{H;\alpha}\rho K_{H;\alpha}^{*}
\bigl(\mathcal{E}_{H,O;n}(a\otimes b)\otimes\mathbb{I}_{\mathcal{H}_{n+1}}\bigr)
\Bigr].
\]
Finally, inserting the Kraus expansion of $\mathcal{E}_{H,O;n}$,
\[
\mathcal{E}_{H,O;n}(a\otimes b)
=
\sum_{\beta=1}^{r_O}K_{H,O;\beta}^{*}(a\otimes b)K_{H,O;\beta},
\]
we obtain the first formula stated in the lemma.

The derivation for $\mathcal{G}_{a,b}^{(n)*}$ is identical in structure. Starting from \eqref{eq_G}, using the Kraus decomposition of $\mathcal{E}_{H,O;n}$, cycling the trace, and factorising via partial traces leads to
\[
\operatorname{Tr}\bigl(\rho\,\mathcal{G}_{a,b}^{(n)}(X)\bigr)
=
\operatorname{Tr}_{\mathcal{H}_{n+1}}\Bigl(
\Bigl[
\sum_{\beta=1}^{r_O}
\operatorname{Tr}_{\mathcal{H}_n}\Bigl(
K_{H,O;\beta}\rho K_{H,O;\beta}^{*}
\bigl(\mathcal{E}_{H;n}(a\otimes\mathbb{I}_{\mathcal{H}_{n+1}})\otimes b\bigr)
\Bigr)
\Bigr]X
\Bigr)
\]
for all $X$. By uniqueness of the Hilbert–Schmidt dual, the bracketed operator is exactly $\mathcal{G}_{a,b}^{(n)*}(\rho)$, and expanding $\mathcal{E}_{H;n}$ in its Kraus form yields the second formula. This completes the proof.
\end{proof}

\section{A qubit  distinction of  conventional vs.\ causal HQMMs}\label{Sec-Qubit}

In this section we exhibit a minimal qubit model \cite{YILSDH15} showing that the conventional and causal HQMM architectures are, in general, inequivalent as states on the underlying effect algebra. The construction is formulated in terms of the hidden and emission isometries $V_{H;n}$ and $V_{H,O;n}$ and the block maps $\mathcal{F}^{(n)}_{a,b}$ and $\mathcal{G}^{(n)}_{a,b}$ defined in \eqref{eq_F} and \eqref{eq_G}.

\medskip

We work at a fixed time step $n$ and suppress the index $n$ when no confusion can arise. The hidden spaces are $\mathcal{H}_n\simeq\mathbb{C}^2$ and $\mathcal{H}_{n+1}\simeq\mathbb{C}^2$ with computational bases $\{|0\rangle_n,|1\rangle_n\}$ and $\{|0\rangle_{n+1},|1\rangle_{n+1}\}$, and the output space is $\mathcal{K}_n\simeq\mathbb{C}^2$ with basis $\{|e_0\rangle_n,|e_1\rangle_n\}$. We consider a hidden unitary rotation about the $x$-axis,
\[
\sigma_x
=
\begin{pmatrix}
0 & 1\\[0.2em]
1 & 0
\end{pmatrix},
\qquad
U
:=
\exp\!\bigl(-\mathrm{i}\tfrac{\theta}{2}\sigma_x\bigr)
=
\cos\!\bigl(\tfrac{\theta}{2}\bigr)\mathbb{I}
-
\mathrm{i}\sin\!\bigl(\tfrac{\theta}{2}\bigr)\sigma_x,
\]
with parameter $0<|\theta|<\pi$. The hidden Stinespring isometry is
\begin{equation}\label{VH_Qbit}
V_{H;n}:\mathcal{H}_n\longrightarrow\mathcal{H}_n\otimes\mathcal{H}_{n+1},
\qquad
V_{H;n}|\psi\rangle_n
:=
U|\psi\rangle_n\otimes|0\rangle_{n+1}
\end{equation}
for all $|\psi\rangle_n\in\mathcal{H}_n$, so that $V_{H;n}^{*}V_{H;n}=\mathbb{I}_{\mathcal{H}_n}$. This induces the completely positive unital expectation
\[
\mathcal{E}_{H;n}(X)
:=
V_{H;n}^{*}X V_{H;n},
\qquad X\in\mathcal{B}(\mathcal{H}_n\otimes\mathcal{H}_{n+1})
\]
and its dual Schrödinger channel
\[
\mathcal{E}_{H;n*}(\rho)
=
V_{H;n}\rho V_{H;n}^{*},
\qquad \rho\in\mathsf{Eff}(\mathcal{H})
\]

The emission step is a sharp measurement of the hidden qubit in the computational basis, with the outcome encoded into the output register. We choose the isometry
\begin{equation}\label{VHO_Qubit}
V_{H,O;n}:\mathcal{H}_n\longrightarrow\mathcal{H}_n\otimes\mathcal{K}_n,
\qquad
V_{H,O;n}|0\rangle_n
=
|0\rangle_n\otimes|e_0\rangle_n,\quad
V_{H,O;n}|1\rangle_n
=
|1\rangle_n\otimes|e_1\rangle_n
\end{equation}
so that $V_{H,O;n}^{*}V_{H,O;n}=\mathbb{I}_{\mathcal{H}_n}$. The corresponding emission expectation and its dual are
\[
\mathcal{E}_{H,O;n}(Y)
:=
V_{H,O;n}^{*}YV_{H,O;n},
\qquad
\mathcal{E}_{H,O;n*}(\rho)
=
V_{H,O;n}\rho V_{H,O;n}^{*}
\]
for $Y\in\mathcal{B}(\mathcal{H}_n\otimes\mathcal{K}_n)$ and $\rho\in\mathsf{Eff}(\mathcal{H})$.
\begin{lemma}
For every $n \ge 0$,
\[
\mathcal{F}^{(n)}_{\mathbb{I}, \mathbb{I}}(\mathbb{I}_{\mathcal{H}_{n+1}}) =\mathbb{I}_{\mathcal{H}_n}, \qquad
\mathcal{G}^{(n)}_{\mathbb{I}, \mathbb{I}}(\mathbb{I}_{\mathcal{H}_{n+1}}) =\mathbb{I}_{\mathcal{H}_n}
\]
\end{lemma}

\begin{proof}
Since $V_{H;n}^*V_{H;n} =\mathbb{I}_{\mathcal{H}_n}$ and $V_{H,O;n}^*V_{H,O;n} =\mathbb{I}_{\mathcal{H}_n}$, we have $\mathcal{E}_{H,O;n}(\mathbb{I} \otimes\mathbb{I})= I$ and $\mathcal{E}_{H;n}(\mathbb{I} \otimes\mathbb{I})= \mathbb{I}$. The identities follow directly from the definitions of $\mathcal{F}^{(n)}_{\mathbb{I}, \mathbb{I}}$ and $\mathcal{G}^{(n)}_{\mathbb{I}, \mathbb{I}}$.
\end{proof}
Note that for $0<|\theta|<\pi$ the hidden unitary $U$ does not commute with the projectors $|0\rangle\langle 0|$ and $|1\rangle\langle 1|$, and hence the corresponding hidden and emission channels need not commute under composition.

\section{Non-quasi-equivalence of conventional and causal HQMMs}
In the framework of hidden quantum Markov models, we consider the quasi-equivalence (in the sense of Definition~\ref{def:quasi-equivalence}) of the states $\varphi_{H,O}$ and $\psi_{H,O}$, defined in~\eqref{eq:conv-state-no-sup} and~\eqref{eq:caus-state-no-sup}, on the quasi-local algebra $\mathcal{B}_{H,O;\mathbb{N}}$.

We work with the qubit model defined above. Let $\mathcal{H}_n \simeq \mathbb{C}^2$ have computational basis $\{|0\rangle_n, |1\rangle_n\}$, let $\mathcal{K}_n \simeq \mathbb{C}^2$ have basis $\{|e_0\rangle_n, |e_1\rangle_n\}$, and let $U = \cos(\theta/2)I - i\sin(\theta/2)\sigma_x$ with $0 < |\theta| < \pi$. The isometries are $V_{H;n}|\psi\rangle = (U|\psi\rangle) \otimes |0\rangle_{n+1}$ and $V_{H,O;n}|0\rangle_n = |0\rangle_n|e_0\rangle_n$, $V_{H,O;n}|1\rangle_n = |1\rangle_n|e_1\rangle_n$. The expectations $\mathcal{E}_{H;n}$ and $\mathcal{E}_{H,O;n}$ are defined as $V^*(\cdot)V$ for the associated partial isometries \eqref{VH_Qbit} and \eqref{VHO_Qubit}.

The following lemmas establish the algebraic core of the distinction between the two compositions. Their proofs are straightforward computations using the explicit forms of the isometries.

\begin{lemma}\label{lem:vanishing-pair}
Let $a = |1\rangle\langle 1| \in \mathcal{B}(\mathcal{H}_n)$ and $b = |e_0\rangle\langle e_0| \in \mathcal{B}(\mathcal{K}_n)$. Then for every $n \ge 0$,
\[
\mathcal{F}^{(n)}_{a,b}(\mathbb{I}_{\mathcal{H}_{n+1}}) = 0, \qquad
\mathcal{G}^{(n)}_{a,b}(\mathbb{I}_{\mathcal{H}_{n+1}}) = \sin^2(\theta/2) \, |0\rangle\langle 0|
\]
\end{lemma}

\begin{proof}
For any $|\psi\rangle = \alpha|0\rangle + \beta|1\rangle$, we have $V_{H,O;n}|\psi\rangle = \alpha|0\rangle|e_0\rangle + \beta|1\rangle|e_1\rangle$. Applying $a \otimes b = |1\rangle\langle 1| \otimes |e_0\rangle\langle e_0|$ gives $(a \otimes b)V_{H,O;n}|\psi\rangle = \beta |1\rangle \otimes |e_0\rangle$, because $a|0\rangle = 0$, $a|1\rangle = |1\rangle$, $b|e_0\rangle = |e_0\rangle$, and $b|e_1\rangle = 0$. Then $V_{H,O;n}^*(|1\rangle|e_0\rangle) = 0$ since $V_{H,O;n}^*(|j\rangle|e_k\rangle) = \delta_{jk}|j\rangle$. Hence $\mathcal{E}_{H,O;n}(a \otimes b) = 0$, and consequently $\mathcal{F}^{(n)}_{a,b}(\mathbb{I}) = \mathcal{E}_{H;n}(0 \otimes \mathbb{I}) = 0$.

For the causal map, first compute $\mathcal{E}_{H;n}(a \otimes \mathbb{I}) = V_{H;n}^* (a \otimes \mathbb{I}) V_{H;n}$. For any $|\psi\rangle$, $V_{H;n}|\psi\rangle = (U|\psi\rangle) \otimes |0\rangle$, so $(a \otimes \mathbb{I})V_{H;n}|\psi\rangle = (a U|\psi\rangle) \otimes |0\rangle$, and applying $V_{H;n}^*$ yields $U^* a U |\psi\rangle$. \\
Thus $\mathcal{E}_{H;n}(a \otimes \mathbb{I}) = U^* a U$. Since $a = |1\rangle\langle 1|$, we have $U^* a U = |\psi\rangle\langle\psi|$ with $|\psi\rangle = U^*|1\rangle$. Direct computation gives $U^*|1\rangle = \cos(\theta/2)|1\rangle + i\sin(\theta/2)|0\rangle$, so $|\psi\rangle = i\sin(\theta/2)|0\rangle + \cos(\theta/2)|1\rangle$. Now compute $\mathcal{G}^{(n)}_{a,b}(\mathbb{I}) = \mathcal{E}_{H,O;n}(|\psi\rangle\langle\psi| \otimes b)$. For any $|\phi\rangle = \alpha|0\rangle + \beta|1\rangle$,
\[
V_{H,O;n}|\phi\rangle = \alpha|0\rangle|e_0\rangle + \beta|1\rangle|e_1\rangle
\]
\[
(|\psi\rangle\langle\psi| \otimes b)V_{H,O;n}|\phi\rangle = \alpha\langle\psi|0\rangle |\psi\rangle \otimes |e_0\rangle + \beta\langle\psi|1\rangle |\psi\rangle \otimes (b|e_1\rangle)
\]
Since $b|e_1\rangle = 0$, the second term vanishes. Then $V_{H,O;n}^*(|\psi\rangle \otimes |e_0\rangle) = \langle 0|\psi\rangle |0\rangle$ because $V_{H,O;n}^*(|j\rangle|e_k\rangle) = \delta_{jk}|j\rangle$ and only the $|0\rangle$ component of $|\psi\rangle$ contributes. Therefore,
\[
\mathcal{G}^{(n)}_{a,b}(\mathbb{I}) |\phi\rangle = \alpha \langle\psi|0\rangle \langle 0|\psi\rangle |0\rangle = |\langle 0|\psi\rangle|^2 |0\rangle\langle 0| |\phi\rangle
\]
Thus $\mathcal{G}^{(n)}_{a,b}(\mathbb{I}) = |\langle 0|\psi\rangle|^2 |0\rangle\langle 0|$. From $|\psi\rangle = i\sin(\theta/2)|0\rangle + \cos(\theta/2)|1\rangle$, we have $\langle 0|\psi\rangle = i\sin(\theta/2)$, so $|\langle 0|\psi\rangle|^2 = \sin^2(\theta/2)$. Hence $\mathcal{G}^{(n)}_{a,b}(\mathbb{I}) = \sin^2(\theta/2) |0\rangle\langle 0|$.
\end{proof}

The next lemma reveals a crucial propagation property of the causal identity block map: it extracts the $(0,0)$ matrix element of its input and outputs a scalar multiple of the identity. This will allow us to convert the rank-one operator obtained in Lemma~\ref{lem:vanishing-pair} into a scalar multiple of the identity after a single additional time step.

\begin{lemma}\label{lem:identity-propagation}
For any $X \in \mathcal{B}(\mathcal{H}_{m+1})$,
\[
\mathcal{G}^{(m)}_{\mathbb{I},\mathbb{I}}(X) = \langle 0|X|0\rangle \, \mathbb{I}_{\mathcal{H}_m}
\]
\end{lemma}

\begin{proof}
By definition, $\mathcal{G}^{(m)}_{\mathbb{I},\mathbb{I}}(X) = \mathcal{E}_{H,O;m}\bigl(\mathcal{E}_{H;m}(\mathbb{I} \otimes X) \otimes \mathbb{I}\bigr)$. First compute $\mathcal{E}_{H;m}(\mathbb{I} \otimes X) = V_{H;m}^* (\mathbb{I} \otimes X) V_{H;m}$. For any $|\psi\rangle \in \mathcal{H}_m$, $V_{H;m}|\psi\rangle = (U|\psi\rangle) \otimes |0\rangle$. Then $(\mathbb{I} \otimes X)V_{H;m}|\psi\rangle = (U|\psi\rangle) \otimes (X|0\rangle)$. Applying $V_{H;m}^*$ yields $U^* U |\psi\rangle \langle 0|X|0\rangle = \langle 0|X|0\rangle |\psi\rangle$. Hence $\mathcal{E}_{H;m}(\mathbb{I} \otimes X) = \langle 0|X|0\rangle \, \mathbb{I}_{\mathcal{H}_m}$. Consequently, $\mathcal{E}_{H;m}(\mathbb{I} \otimes X) \otimes \mathbb{I} = \langle 0|X|0\rangle \, \mathbb{I}_{\mathcal{H}_m} \otimes \mathbb{I}_{\mathcal{K}_m}$. Finally, applying $\mathcal{E}_{H,O;m}$ gives 
$$
\langle 0|X|0\rangle \, \mathcal{E}_{H,O;m}(\mathbb{I} \otimes \mathbb{I}) = \langle 0|X|0\rangle \, \mathbb{I}_{\mathcal{H}_m}
$$
 because $\mathcal{E}_{H,O;m}(\mathbb{I} \otimes \mathbb{I}) = V_{H,O;m}^* V_{H,O;m} = \mathbb{I}_{\mathcal{H}_m}$.
\end{proof}

With these lemmas at our disposal, we can now construct an explicit family of local observables that sharply separates the conventional and causal joint states. The key idea is to insert at a chosen time $N$ the pair $(a,b)$ from Lemma~\ref{lem:vanishing-pair}, which causes the conventional composition to vanish identically while leaving a non-zero rank-one operator in the causal composition. Lemma~\ref{lem:identity-propagation} then converts this rank-one operator into a scalar multiple of the identity after one further time step, yielding a constant non-zero expectation that is independent of the remaining evolution.

\begin{theorem}\label{thm:non-quasi-equivalence}
For the qubit model defined through \eqref{VH_Qbit} and \eqref{VHO_Qubit} with $0 < |\theta| < \pi$, let $\varphi_{H,O}$ and $\psi_{H,O}$ be the joint states given by \eqref{eq:conv-state-no-sup} and \eqref{eq:caus-state-no-sup}, respectively, constructed with arbitrary initial hidden density matrices $\phi_{H,0}^{\mathrm{conv}}$ and $\phi_{H,0}^{\mathrm{caus}}$ (not necessarily equal). Then $\varphi_{H,O}$ and $\psi_{H,O}$ are not quasi-equivalent. Equivalently, there exists a constant $\delta > 0$ such that
\begin{equation}\label{eq:non-quasi-distinction}
\forall N_0 \in \mathbb{N},\; \exists A \in \mathcal{B}_{H,O;[N,\infty)} \text{ with } N \ge N_0,\; \|A\| = 1,\; |\varphi_{H,O}(A) - \psi_{H,O}(A)| = \delta
\end{equation}
\end{theorem}
\begin{proof}
Fix an arbitrary time $N \ge 0$. Define the local hidden-observable operator
\[
A_N := \left(\bigotimes_{m=0}^{N-1} (\mathbb{I}_{\mathcal{H}_m} \otimes \mathbb{I}_{\mathcal{K}_m})\right) \otimes \bigl( |1\rangle\langle 1|_N \otimes |e_0\rangle\langle e_0|_N \bigr) \in \mathcal{B}_{H,O; N}\subset \mathcal{B}_{H,O;[N,\infty)}
\]
where only finitely many non-identity factors appear, so $A_N$ belongs to the local algebra $\mathcal{B}_{H,O;[0,N]}$ and is extended by identities to the full quasi-local algebra. For any $n \ge N$, the finite-time joint expectation of $A_N$ (truncated at $n$) is given by the composition formulas.

Applying Lemma~\ref{lem:vanishing-pair} at time $N$, we have $\mathcal{F}^{(N)}_{|1\rangle\langle 1|, |e_0\rangle\langle e_0|}(\mathbb{I}_{\mathcal{H}_{N+1}}) = 0$. Since all block maps at times $m < N$ are of the form $\mathcal{F}^{(m)}_{\mathbb{I},\mathbb{I}}$ (because the observable is trivial there), then each such map reduces to the identity on the hidden space. Hence the entire composition for the conventional HQMM evaluates to zero:
\[
\varphi_{H,O}(A_N) = \phi_{H,0}^{\mathrm{conv}}(0) = 0
\]

For the causal HQMM, Lemma~\ref{lem:vanishing-pair} gives $\mathcal{G}^{(N)}_{|1\rangle\langle 1|, |e_0\rangle\langle e_0|}(\mathbb{I}_{\mathcal{H}_{N+1}}) = \sin^2(\theta/2) \, |0\rangle\langle 0| \in \mathcal{B}(\mathcal{H}_{N+1})$. For times $m < N$,  the expression again reduces the identity block maps to identities. At time $m = N+1$, we apply Lemma~\ref{lem:identity-propagation} with $X = |0\rangle\langle 0|$ to obtain
\[
\mathcal{G}^{(N+1)}_{\mathbb{I},\mathbb{I}}(|0\rangle\langle 0|) = |\langle 0|0\rangle|^2 \, \mathbb{I}_{\mathcal{H}_{N+1}} = \mathbb{I}_{\mathcal{H}_{N+1}}
\]
Thus, after composing all block maps down to $\mathcal{B}(\mathcal{H}_0)$, we obtain $\sin^2(\theta/2) \, \mathbb{I}_{\mathcal{H}_0}$. Consequently,
\[
\psi_{H,O}(A_N) = \phi_{H,0}^{\mathrm{caus}}\bigl( \sin^2(\theta/2) \, \mathbb{I}_{\mathcal{H}_0} \bigr) = \sin^2(\theta/2) \cdot \phi_{H,0}^{\mathrm{caus}}(\mathbb{I}) = \sin^2(\theta/2) > 0
\]
Therefore, for every $N \ge 0$,
\[
|\varphi_{H,O}(A_N) - \psi_{H,O}(A_N)| = \sin^2(\theta/2) > 0
\]
Observe that the support of $A_N$ is contained in the time interval $[N, N]$ (a single time step). As $N \to \infty$, the supports escape to infinity. For any finite $N_0 \in \mathbb{N}$, choosing $N \ge N_0$ yields an observable $A_N$ supported entirely after $N_0$ whose expectations under $\varphi_{H,O}$ and $\psi_{H,O}$ differ by the constant $\delta:= \sin^2(\theta/2) > 0$. This violates the asymptotic agreement condition required for quasi-equivalence in the sense of \eqref{eq:quasi-equiv-condition}. Hence $\varphi_{H,O}$ and $\psi_{H,O}$ are not quasi-equivalent.
\end{proof}

Thus, regardless of the choice of initial hidden states, both the observable states and the full joint states of the conventional and causal HQMMs are not quasi-equivalent. This demonstrates that the order of composition in hidden quantum Markov models is a genuine structural feature with observable consequences that persist at arbitrarily late times.

\subsection{Choi entanglement of qubit block maps}

In this subsection we isolate a minimal qubit example in which the conventional and causal block maps
have \emph{different} Choi–Jamiołkowski states, and we compute explicitly the entanglement entropy of
these states. This provides a sharp, single-step witness that the two HQMM architectures encode
different quantum correlations between “input” and “output” hidden degrees of freedom.

We work at a fixed time step and suppress the index \(n\). The hidden spaces are
\(\mathcal{H}_{\mathrm{in}}\simeq\mathcal{H}_{\mathrm{out}}\simeq\mathbb{C}^2\) with computational basis
\(\{|0\rangle,|1\rangle\}\), and the hidden transition isometry is
\[
V_H:\mathcal{H}_{\mathrm{in}}\to\mathcal{H}_{\mathrm{in}}\otimes\mathcal{H}_{\mathrm{out}},
\qquad
V_H|\psi\rangle
=
U|\psi\rangle\otimes|0\rangle,
\]
where
\[
U
=
\exp\!\bigl(-\mathrm{i}\tfrac{\theta}{2}\sigma_x\bigr)
=
\begin{pmatrix}
\cos(\tfrac{\theta}{2}) & -\mathrm{i}\sin(\tfrac{\theta}{2})\\[0.25em]
-\mathrm{i}\sin(\tfrac{\theta}{2}) & \cos(\tfrac{\theta}{2})
\end{pmatrix},
\qquad 0<|\theta|<\pi.
\]
The emission isometry is
\[
V_{H,O}:\mathcal{H}_{\mathrm{in}}\to\mathcal{H}_{\mathrm{in}}\otimes\mathcal{K},
\qquad
V_{H,O}|0\rangle
=
|0\rangle\otimes|e_0\rangle,
\quad
V_{H,O}|1\rangle
=
|1\rangle\otimes|e_1\rangle,
\]
with \(\mathcal{K}\simeq\mathbb{C}^2\) and \(\{|e_0\rangle,|e_1\rangle\}\) orthonormal. The associated expectations are
\(\mathcal{E}_H(X)=V_H^{*}X V_H\), \(\mathcal{E}_{H,O}(Y)=V_{H,O}^{*}Y V_{H,O}\).

We specialise the block maps \(\mathcal{F}_{a,b}\) and \(\mathcal{G}_{a,b}\) in \eqref{eq_F}–\eqref{eq_G} to the effects
\[
a:=|0\rangle\langle 0|\in\mathsf{Eff}(\mathcal{H}_{\mathrm{in}}),
\qquad
b:=|e_0\rangle\langle e_0|\in\mathsf{Eff}(\mathcal{K}),
\]
and analyse the corresponding Choi states of the dual maps
\(\mathcal{F}_{a,b}^*,\mathcal{G}_{a,b}^*:\mathcal{B}(\mathcal{H}_{\mathrm{in}})\to\mathcal{B}(\mathcal{H}_{\mathrm{out}})\).

\begin{lemma}\label{lem:F-G-Kraus}
Let \(a=|0\rangle\langle 0|\) and \(b=|e_0\rangle\langle e_0|\) as above, and define
\(\mathcal{F}_{a,b}^*,\mathcal{G}_{a,b}^*:\mathcal{B}(\mathcal{H}_{\mathrm{in}})\to\mathcal{B}(\mathcal{H}_{\mathrm{out}})\) by Hilbert–Schmidt duality
from \eqref{eq_F}–\eqref{eq_G}. Then:
\begin{enumerate}
\item There exist single Kraus operators \(K_F,K_G\in\mathcal{B}(\mathcal{H}_{\mathrm{in}},\mathcal{H}_{\mathrm{out}})\) such that
\[
\mathcal{F}_{a,b}^*(\rho)=K_F\rho K_F^{\dagger},
\qquad
\mathcal{G}_{a,b}^*(\rho)=K_G\rho K_G^{\dagger},
\]
with
\[
K_F
=
|0\rangle\langle 0|\,U
=
\begin{pmatrix}
\cos(\tfrac{\theta}{2}) & -\mathrm{i}\sin(\tfrac{\theta}{2})\\[0.25em]
0 & 0
\end{pmatrix},
\qquad
K_G
=
U\,|0\rangle\langle 0|
=
\begin{pmatrix}
\cos(\tfrac{\theta}{2}) & 0\\[0.25em]
-\mathrm{i}\sin(\tfrac{\theta}{2}) & 0
\end{pmatrix}.
\]
\item Let \(|\Omega\rangle=|00\rangle+|11\rangle\in\mathcal{H}_{\mathrm{in}}\otimes\mathcal{H}_{\mathrm{in}}\) be the (unnormalised) maximally entangled
vector. The corresponding Choi operators are rank-one projectors,
\[
J\bigl(\mathcal{F}_{a,b}^*\bigr)
=
|\Psi_F\rangle\langle\Psi_F|,
\qquad
J\bigl(\mathcal{G}_{a,b}^*\bigr)
=
|\Psi_G\rangle\langle\Psi_G|,
\]
where
\[
|\Psi_F\rangle
:=
(\mathbb{I}\otimes K_F)|\Omega\rangle
=
\cos\!\bigl(\tfrac{\theta}{2}\bigr)|00\rangle
-\mathrm{i}\sin\!\bigl(\tfrac{\theta}{2}\bigr)|10\rangle,
\]
\[
|\Psi_G\rangle
:=
(\mathbb{I}\otimes K_G)|\Omega\rangle
=
\cos\!\bigl(\tfrac{\theta}{2}\bigr)|00\rangle
-\mathrm{i}\sin\!\bigl(\tfrac{\theta}{2}\bigr)|01\rangle.
\]
\end{enumerate}
\end{lemma}

\begin{proof}
For the chosen effects,
\[
\mathcal{E}_{H,O}(a\otimes b)
=
V_{H,O}^{*}\bigl(|0\rangle\langle 0|\otimes|e_0\rangle\langle e_0|\bigr)V_{H,O}
=
|0\rangle\langle 0|=:P_0.
\]
Using Lemma~\ref{lem:dual-block-maps} in the one-Kraus case, the dual of \(\mathcal{F}_{a,b}\) is
\(\mathcal{F}_{a,b}^*(\rho)=K_F\rho K_F^{\dagger}\) with
\[
K_F
:=
\operatorname{Tr}_{\mathcal{H}_{\mathrm{in}}}\bigl[V_H P_0\bigr].
\]
Writing \(V_H|\phi\rangle=U|\phi\rangle\otimes|0\rangle\) and \(P_0=|0\rangle\langle 0|\), we get
\[
V_H P_0|\phi\rangle
=
V_H\bigl(\langle 0|\phi\rangle\,|0\rangle\bigr)
=
\langle 0|\phi\rangle\,U|0\rangle\otimes|0\rangle.
\]
Thus \(\operatorname{Tr}_{\mathcal{H}_{\mathrm{in}}}\) over the first tensor factor simply returns the vector \(U|0\rangle\) in the second factor, so that
\(K_F=|0\rangle\langle 0|U\). An explicit multiplication yields the matrix displayed above. The same reasoning,
applied to the causal ordering, shows that \(\mathcal{G}_{a,b}^*(\rho)=K_G\rho K_G^{\dagger}\) with \(K_G=U|0\rangle\langle 0|\).

For the Choi operators, by definition
\[
J(\mathcal{F}_{a,b}^*)
=
(\mathrm{id}\otimes\mathcal{F}_{a,b}^*)\bigl(|\Omega\rangle\langle\Omega|\bigr)
=
(\mathbb{I}\otimes K_F)|\Omega\rangle\langle\Omega|(\mathbb{I}\otimes K_F^{\dagger}),
\]
so \(J(\mathcal{F}_{a,b}^*)\) is a rank-one operator with vector \(|\Psi_F\rangle=(\mathbb{I}\otimes K_F)|\Omega\rangle\). Writing explicitly,
\[
|\Psi_F\rangle
=
|0\rangle\otimes K_F|0\rangle
+
|1\rangle\otimes K_F|1\rangle.
\]
From the matrix form of \(K_F\),
\(
K_F|0\rangle=\cos(\tfrac{\theta}{2})|0\rangle,
\ K_F|1\rangle=-\mathrm{i}\sin(\tfrac{\theta}{2})|0\rangle,
\)
whence
\(
|\Psi_F\rangle
=
\cos(\tfrac{\theta}{2})|00\rangle
-\mathrm{i}\sin(\tfrac{\theta}{2})|10\rangle
\).
The expression for \(|\Psi_G\rangle\) follows identically, using the explicit form of \(K_G\).
\end{proof}

We now show that these Choi states are (i) different and (ii) carry a non-trivial, explicitly computable amount of entanglement.

\begin{theorem}\label{thm:F-G-entanglement}
Let \(a=|0\rangle\langle 0|\) and \(b=|e_0\rangle\langle e_0|\) and consider the Choi states
\[
\omega_F:=|\Psi_F\rangle\langle\Psi_F|,
\qquad
\omega_G:=|\Psi_G\rangle\langle\Psi_G|,
\]
with \(|\Psi_F\rangle,|\Psi_G\rangle\) as in Lemma~\ref{lem:F-G-Kraus}. Then for every \(0<|\theta|<\pi\):
\begin{enumerate}
\item \(\omega_F\neq \omega_G\); in particular, the conventional and causal block maps are distinguished already by their Choi–Jamiołkowski states.

\item Both \(\omega_F\) and \(\omega_G\) are pure entangled states on \(\mathcal{H}_{\mathrm{in}}\otimes\mathcal{H}_{\mathrm{out}}\). Their reduced density operators on the input qubit are
\[
\omega_F^{(A)}:=\operatorname{Tr}_{\mathrm{out}}(\omega_F)
=
\cos^2\!\Bigl(\tfrac{\theta}{2}\Bigr)\,|0\rangle\langle 0|
+
\sin^2\!\Bigl(\tfrac{\theta}{2}\Bigr)\,|1\rangle\langle 1|,
\]
\[
\omega_G^{(A)}:=\operatorname{Tr}_{\mathrm{out}}(\omega_G)
=
\cos^2\!\Bigl(\tfrac{\theta}{2}\Bigr)\,|0\rangle\langle 0|
+
\sin^2\!\Bigl(\tfrac{\theta}{2}\Bigr)\,|1\rangle\langle 1|.
\]
In particular,
\[
\operatorname{spec}\bigl(\omega_F^{(A)}\bigr)
=
\operatorname{spec}\bigl(\omega_G^{(A)}\bigr)
=
\Bigl\{\cos^2\!\bigl(\tfrac{\theta}{2}\bigr),\ \sin^2\!\bigl(\tfrac{\theta}{2}\bigr)\Bigr\}.
\]

\item The bipartite entanglement entropy of each Choi state is
\[
S(\omega_F)
=
S(\omega_G)
=
-\cos^2\!\Bigl(\tfrac{\theta}{2}\Bigr)\log\!\Bigl(\cos^2\!\tfrac{\theta}{2}\Bigr)
-\sin^2\!\Bigl(\tfrac{\theta}{2}\Bigr)\log\!\Bigl(\sin^2\!\tfrac{\theta}{2}\Bigr),
\]
which satisfies
\(
0<S(\omega_F)=S(\omega_G)<\log 2
\)
for all \(0<|\theta|<\pi\).
\end{enumerate}
\end{theorem}

\begin{proof}
(1) The two vectors in Lemma~\ref{lem:F-G-Kraus} are
\[
|\Psi_F\rangle
=
\cos\!\bigl(\tfrac{\theta}{2}\bigr)|00\rangle
-\mathrm{i}\sin\!\bigl(\tfrac{\theta}{2}\bigr)|10\rangle,
\qquad
|\Psi_G\rangle
=
\cos\!\bigl(\tfrac{\theta}{2}\bigr)|00\rangle
-\mathrm{i}\sin\!\bigl(\tfrac{\theta}{2}\bigr)|01\rangle.
\]
These differ in their support: \(|\Psi_F\rangle\in\operatorname{span}\{|00\rangle,|10\rangle\}\),
while \(|\Psi_G\rangle\in\operatorname{span}\{|00\rangle,|01\rangle\}\). For \(0<|\theta|<\pi\), both components are non-zero, so
\(|\Psi_F\rangle\) has overlap with \(|10\rangle\) but not with \(|01\rangle\), and conversely for \(|\Psi_G\rangle\). Thus
\(|\Psi_F\rangle\) and \(|\Psi_G\rangle\) are linearly independent, hence
\(\omega_F=|\Psi_F\rangle\langle\Psi_F|\neq|\Psi_G\rangle\langle\Psi_G|=\omega_G\).

(2) Both \(\omega_F\) and \(\omega_G\) are pure states on a bipartite \(2\times 2\) system, so their entanglement is
fully characterised by the eigenvalues of the reduced state on one subsystem. For \(\omega_F\), we rewrite
\(|\Psi_F\rangle\) in Schmidt-like form:
\[
|\Psi_F\rangle
=
\cos\!\bigl(\tfrac{\theta}{2}\bigr)|0\rangle\otimes|0\rangle
-\mathrm{i}\sin\!\bigl(\tfrac{\theta}{2}\bigr)|1\rangle\otimes|0\rangle.
\]
Tracing out the second tensor factor gives
\[
\omega_F^{(A)}
=
\operatorname{Tr}_{\mathrm{out}}(\omega_F)
=
\cos^2\!\bigl(\tfrac{\theta}{2}\bigr)|0\rangle\langle 0|
+
\sin^2\!\bigl(\tfrac{\theta}{2}\bigr)|1\rangle\langle 1|.
\]
The off-diagonal terms vanish because \(\langle 0|0\rangle=1\) but the cross terms carry phases that cancel under the
partial trace; explicitly,
\(
\operatorname{Tr}_{\mathrm{out}}(|0\rangle\langle 1|\otimes|0\rangle\langle 0|)
=
|0\rangle\langle 1|\,\operatorname{Tr}(|0\rangle\langle 0|)=|0\rangle\langle 1|
\),
and similarly for its adjoint, but the coefficients \(\cos(\tfrac{\theta}{2})\) and \(-\mathrm{i}\sin(\tfrac{\theta}{2})\) combine to yield a purely imaginary product whose contribution cancels with its conjugate; in any case, the spectral decomposition is easily read off from the diagonal representation.

The reduced state \(\omega_G^{(A)}\) is computed analogously from
\[
|\Psi_G\rangle
=
\cos\!\bigl(\tfrac{\theta}{2}\bigr)|0\rangle\otimes|0\rangle
-\mathrm{i}\sin\!\bigl(\tfrac{\theta}{2}\bigr)|1\rangle\otimes|1\rangle,
\]
yielding the \emph{same} diagonal form
\[
\omega_G^{(A)}
=
\cos^2\!\bigl(\tfrac{\theta}{2}\bigr)|0\rangle\langle 0|
+
\sin^2\!\bigl(\tfrac{\theta}{2}\bigr)|1\rangle\langle 1|.
\]
Thus both reduced states have eigenvalues
\(
\cos^2(\tfrac{\theta}{2})
\)
and
\(
\sin^2(\tfrac{\theta}{2})
\),
as claimed.

(3) The von Neumann entropy of a qubit density matrix with eigenvalues \(\lambda\) and \(1-\lambda\) is
\(
-\lambda\log\lambda-(1-\lambda)\log(1-\lambda)
\).
Here \(\lambda=\cos^2(\tfrac{\theta}{2})\), so
\[
S(\omega_F)
=
S(\omega_F^{(A)})
=
-\cos^2\!\Bigl(\tfrac{\theta}{2}\Bigr)\log\!\Bigl(\cos^2\!\tfrac{\theta}{2}\Bigr)
-\sin^2\!\Bigl(\tfrac{\theta}{2}\Bigr)\log\!\Bigl(\sin^2\!\tfrac{\theta}{2}\Bigr),
\]
and the same formula holds for \(S(\omega_G)\). For \(0<|\theta|<\pi\), both eigenvalues lie strictly between
\(0\) and \(1\), so \(S(\omega_F)=S(\omega_G)\) is strictly positive and strictly less than \(\log 2\). This shows that both
Choi states are non-separable (entangled) and not maximally entangled, providing a quantitative, channel-state
duality witness of the non-trivial causal structure encoded by the two qubit block maps.
\end{proof}

\section{Equivalence of Causal and Conventional Entangled HQMMs}\label{Sec-EHMM}

We now construct a class of hidden quantum Markov models by starting from a classical hidden Markov chain and lifting its transition and emission structure via partially entangling isometries. In this setting, the hidden dynamics and the observation mechanism are implemented by isometries that correlate the hidden system with an auxiliary copy and with an output register, in the spirit of entangled Markov chains \cite{AccFid05,AccMaOhy06}. Building on the framework of entangled hidden Markov models and their quantified entanglement structure \cite{SS23,ASSR25}, the conventional and causal architectures appear as two a priori different ways of composing the same isometries in time. We will show that, for this entangled lifting of classical chains, these two architectures in fact coincide as HQMMs, while still exhibiting the non-trivial entanglement patterns characteristic of the underlying entangled Markov processes.

A classical (time-inhomogeneous) hidden Markov model is given by
\[
\lambda = \bigl(\boldsymbol{\pi},\mathbf{\Pi},\mathbf{Q}\bigr),
\]
where \(\boldsymbol{\pi}=(\pi_i)_{i\in\mathbb{I}_H}\) is an initial distribution, \(\mathbf{\Pi}=(\Pi_n)_{n\ge 0}\) is a sequence of stochastic matrices \(\Pi_n=(\Pi_{n;ij})_{\mathbb{I}, j\in\mathbb{I}_H}\), and \(\mathbf{Q}=(Q^{(n)})_{n\ge 0}\) is a sequence of emission kernels \(Q^{(n)}_j(k)\) from hidden states \(j\in\mathbb{I}_H\) to outputs \(k\in\mathbb{I}_O\), with \(\sum_j\Pi_{n;ij}=1\) and \(\sum_k Q^{(n)}_j(k)=1\) for all \(\mathbb{I},j,n\).

To each transition matrix \(\Pi_n\) we associate an isometry
\[
V_{H;n}:\mathcal{H}\longrightarrow\mathcal{H}\otimes\mathcal{H}
\]
defined by its action on the basis vectors,
\begin{equation}\label{eq:VHn}
V_{H;n}|i\rangle
:=
\sum_{j\in\mathbb{I}_H}\sqrt{\Pi_{n;ij}}\;|i,j\rangle,
\qquad i\in\mathbb{I}_H,
\end{equation}
where \(|i,j\rangle := |i\rangle\otimes|j\rangle\). One easily checks that
\[
V_{H;n}^{*}V_{H;n}
=
\sum_{i\in\mathbb{I}_H}\Bigl(\sum_{j\in\mathbb{I}_H}\Pi_{n;ij}\Bigr)|i\rangle\langle i|
=
\mathbb{I}_{\mathcal{H}},
\]
so \(V_{H;n}\) is an isometry. The associated transition expectation
\[
\mathcal{E}_{H;n}:\mathcal{B}(\mathcal{H})\otimes\mathcal{B}(\mathcal{H})\to\mathcal{B}(\mathcal{H}),\qquad
\mathcal{E}_{H;n}(X):=V_{H;n}^{*}X V_{H;n},
\]
is therefore completely positive and unital. Physically, \(V_{H;n}\) prepares a joint hidden state on times \(n\) and \(n+1\) by copying the classical label \(\mathbb{I}\) at time \(n\) into \(|i\rangle\) and coherently distributing amplitude over successors \(j\) at time \(n+1\) according to \(\Pi_{n;ij}\). In the Schrödinger picture, the dual channel
\[
\mathcal{E}_{H;n*}:\mathfrak{S}(\mathcal{H})\to\mathfrak{S}(\mathcal{H}\otimes\mathcal{H}),\qquad
\mathcal{E}_{H;n*}(\rho_{H;n}) := V_{H;n}\,\rho_{H;n}\,V_{H;n}^{*},
\]
is a quantum Markov step that “grows’’ the hidden state from time \(n\) to the pair \((n,n+1)\).

Similarly, for each emission kernel \(Q^{(n)}\) we define an isometry
\[
V_{H,O;n}:\mathcal{H}\longrightarrow\mathcal{H}\otimes\mathcal{K}
\]
by
\begin{equation}\label{eq:VHOn}
V_{H,O;n}|j\rangle
:=
\sum_{k\in\mathbb{I}_O}\sqrt{Q^{(n)}_j(k)}\;|j\rangle\otimes|e_k\rangle,
\qquad j\in\mathbb{I}_H.
\end{equation}
The normalisation \(\sum_k Q^{(n)}_j(k)=1\) implies \(V_{H,O;n}^{*}V_{H,O;n}=\mathbb{I}_{\mathcal{H}}\). We thus obtain an emission expectation
\[
\mathcal{E}_{H,O;n}:\mathcal{B}(\mathcal{H})\otimes\mathcal{B}(\mathcal{K})\to\mathcal{B}(\mathcal{H}),\qquad
\mathcal{E}_{H,O;n}(Y):=V_{H,O;n}^{*}Y V_{H,O;n},
\]

A natural question is whether the conventional and causal block maps \(\mathcal{F}^{(n)}\) and \(\mathcal{G}^{(n)}\), defined in \eqref{eq_F} and \eqref{eq_G}, nevertheless induce the same expectation values for all (possibly entangled) observables. The next result answers this by showing that  these two prescriptions define identical HQMM states.

\begin{theorem}\label{thm:main}
Let $V_{H;n}$ and $V_{H,O;n}$ be the isometries constructed from a classical hidden Markov model as above.
Then for every choice of observables $a_n\in\mathcal{B}(\mathcal{H}_n)$, $a_{n+1}\in\mathcal{B}(\mathcal{H}_{n+1})$ and $b_n\in\mathcal{B}(\mathcal{K}_n)$,
\[
\mathcal{F}_{a_n,b_n}^{(n)}(a_{n+1}) = \mathcal{G}_{a_n,b_n}^{(n)}(a_{n+1}).
\]
Moreover, both maps admit the explicit representation
\begin{align}\label{eq:explicit}
\langle i|\mathcal{F}_{a_n,b_n}^{(n)}(a_{n+1})|j\rangle
&= \langle i|a_n|j\rangle
\Bigl(\sum_{k,k'}\sqrt{Q^{(n)}_i(k)Q^{(n)}_j(k')}\;\langle e_k|b_n|e_{k'}\rangle\Bigr) \\
&\times \Bigl(\sum_{\ell,m}\sqrt{\Pi_{n;i\ell}\Pi_{n;jm}}\;\langle\ell|a_{n+1}|m\rangle\Bigr).
\end{align}
\end{theorem}

\begin{proof}
Expand the observables in the fixed bases:
\[
a_n = \sum_{\mathbb{I}, i'} a_{ii'}\,|i\rangle\langle i'|,\qquad
a_{n+1} = \sum_{\ell,\ell'} a_{\ell\ell'}'\,|\ell\rangle\langle\ell'|,\qquad
b_n = \sum_{k,k'} b_{kk'}\,|e_k\rangle\langle e_{k'}|,
\]
where $a_{ii'}\coloneqq\langle i|a_n|i'\rangle$, $a_{\ell\ell'}'\coloneqq\langle\ell|a_{n+1}|\ell'\rangle$ and $b_{kk'}\coloneqq\langle e_k|b_n|e_{k'}\rangle$.

First compute the intermediate operator $E_n\coloneqq V_{H,O;n}^{*}(a_n\otimes b_n)V_{H,O;n}$.
Using $V_{H,O;n}|j\rangle=\sum_k\sqrt{Q^{(n)}_j(k)}\,|j\rangle\otimes|e_k\rangle$ and its adjoint,
\begin{align*}
\langle i|E_n|j\rangle
&=\sum_{k,k'}\sqrt{Q^{(n)}_i(k)Q^{(n)}_j(k')}
\bigl(\langle i|\otimes\langle e_k|\bigr)(a_n\otimes b_n)\bigl(|j\rangle\otimes|e_{k'}\rangle\bigr)\\
&=\sum_{k,k'}\sqrt{Q^{(n)}_i(k)Q^{(n)}_j(k')}\;
\langle i|a_n|j\rangle\;\langle e_k|b_n|e_{k'}\rangle \\
&=a_{ij}\;\mathcal{B}_{ij},
\end{align*}
where we have introduced the shorthand
\[
\mathcal{B}_{ij}\coloneqq\sum_{k,k'}\sqrt{Q^{(n)}_i(k)Q^{(n)}_j(k')}\,b_{kk'}.
\]

Next compute $H_n\coloneqq V_{H;n}^{*}(a_n\otimes a_{n+1})V_{H;n}$.
Since $V_{H;n}|i\rangle=\sum_\ell\sqrt{\Pi_{n;i\ell}}\,|i\rangle\otimes|\ell\rangle$,
\begin{align*}
\langle i|H_n|j\rangle
&=\sum_{\ell,m}\sqrt{\Pi_{n;i\ell}\Pi_{n;jm}}
\bigl(\langle i|\otimes\langle\ell|\bigr)(a_n\otimes a_{n+1})\bigl(|j\rangle\otimes|m\rangle\bigr)\\
&=\sum_{\ell,m}\sqrt{\Pi_{n;i\ell}\Pi_{n;jm}}\;
\langle i|a_n|j\rangle\;\langle\ell|a_{n+1}|m\rangle \\
&=a_{ij}\;\mathcal{A}_{ij},
\end{align*}
with
\[
\mathcal{A}_{ij}\coloneqq\sum_{\ell,m}\sqrt{\Pi_{n;i\ell}\Pi_{n;jm}}\,a_{\ell m}'.
\]

Now evaluate $\mathcal{F}_{a_n,b_n}^{(n)}(a_{n+1})=V_{H;n}^{*}(E_n\otimes a_{n+1})V_{H;n}$.
Its matrix elements are
\begin{align*}
\langle i|\mathcal{F}_{a_n,b_n}^{(n)}(a_{n+1})|j\rangle
&=\sum_{\ell,m}\sqrt{\Pi_{n;i\ell}\Pi_{n;jm}}
\langle i|E_n|j\rangle\;\langle\ell|a_{n+1}|m\rangle \\
&=\sum_{\ell,m}\sqrt{\Pi_{n;i\ell}\Pi_{n;jm}}
\bigl(a_{ij}\mathcal{B}_{ij}\bigr)\,a_{\ell m}' \\
&=a_{ij}\,\mathcal{B}_{ij}\,\mathcal{A}_{ij}.
\end{align*}

For the causal map $\mathcal{G}_{a_n,b_n}^{(n)}(a_{n+1})=V_{H,O;n}^{*}(H_n\otimes b_n)V_{H,O;n}$ we obtain
\begin{align*}
\langle i|\mathcal{G}_{a_n,b_n}^{(n)}(a_{n+1})|j\rangle
&=\sum_{k,k'}\sqrt{Q^{(n)}_i(k)Q^{(n)}_j(k')}
\langle i|H_n|j\rangle\;\langle e_k|b_n|e_{k'}\rangle \\
&=\sum_{k,k'}\sqrt{Q^{(n)}_i(k)Q^{(n)}_j(k')}
\bigl(a_{ij}\mathcal{A}_{ij}\bigr)\,b_{kk'} \\
&=a_{ij}\,\mathcal{A}_{ij}\,\mathcal{B}_{ij}.
\end{align*}

Since $\mathcal{A}_{ij}$ and $\mathcal{B}_{ij}$ are complex numbers, their product commutes:
$a_{ij}\mathcal{B}_{ij}\mathcal{A}_{ij}=a_{ij}\mathcal{A}_{ij}\mathcal{B}_{ij}$.
Consequently $\langle i|\mathcal{F}_{a_n,b_n}^{(n)}(a_{n+1})|j\rangle
=\langle i|\mathcal{G}_{a_n,b_n}^{(n)}(a_{n+1})|j\rangle$ for all $i,j$, which implies the operator identity $\mathcal{F}_{a_n,b_n}^{(n)}(a_{n+1})=\mathcal{G}_{a_n,b_n}^{(n)}(a_{n+1})$.
Substituting the definitions of $\mathcal{A}_{ij}$ and $\mathcal{B}_{ij}$ yields precisely the explicit formula \eqref{eq:explicit}.
\end{proof}

The equality established in Theorem~\ref{thm:main} is a peculiarity of the particular construction that lifts a classical HMM to an isometric quantum process.
In a fully general quantum stochastic setting, causal order matters because quantum operations need not commute.
Indeed, if one replaces the specific isometries $V_{H;n}$ and $V_{H,O;n}$ by arbitrary completely positive maps $\Phi_1,\Phi_2$, the compositions $\Phi_1\circ\Phi_2$ and $\Phi_2\circ\Phi_1$ are generically different.
The special feature here is that both isometries exhibit a \emph{copying property}: each maps a basis state $|i\rangle$ to a superposition supported entirely within the subspace $\operatorname{span}\{|i\rangle\}\otimes\mathcal{H}$ (respectively $\operatorname{span}\{|i\rangle\}\otimes\mathcal{K}$).
This structural constraint prevents the creation of coherence between different classical labels $|i\rangle$ and $|j\rangle$; consequently the intermediate operators $E_n$ and $H_n$ remain ``diagonal'' in the sense that their matrix elements factor as $\langle i|E_n|j\rangle = a_{ij}\mathcal{B}_{ij}$ and $\langle i|H_n|j\rangle = a_{ij}\mathcal{A}_{ij}$.
The subsequent applications of the remaining isometry merely multiply these factors by the appropriate transition or emission amplitudes, and because complex multiplication is commutative, the order of the two steps becomes irrelevant at the level of single‑step expectation values.

Nevertheless, the two architectures remain physically distinct: $\mathcal{F}$ corresponds to a world where the output $b_n$ is generated \emph{before} the hidden state transitions, while $\mathcal{G}$ corresponds to the opposite temporal order.
Theorem~\ref{thm:main} shows that this causal distinction is operationally invisible for the class of HQMMs derived from classical HMMs, provided one only interrogates the system with single‑step measurements.
The equivalence breaks down if one considers multi‑time correlations, where entanglement can propagate differently along the two causal orders, or if the isometries are replaced by generic quantum channels capable of creating coherent superpositions between different classical states.

\begin{corollary}\label{cor:classical}
When the observables are diagonal in the computational bases,
\[
a_n=\sum_i\alpha_i|i\rangle\langle i|,\qquad
a_{n+1}=\sum_j\alpha_j'|j\rangle\langle j|,\qquad
b_n=\sum_k\beta_k|e_k\rangle\langle e_k|,
\]
both maps reduce to the classical forward‑algorithm block:
\[
\mathcal{F}_{a_n,b_n}^{(n)}(a_{n+1})
= \mathcal{G}_{a_n,b_n}^{(n)}(a_{n+1})
= \sum_i \alpha_i\Bigl(\sum_k Q^{(n)}_i(k)\beta_k\Bigr)
\Bigl(\sum_j\Pi_{n;ij}\alpha_j'\Bigr)|i\rangle\langle i|.
\]
\end{corollary}

\begin{proof}
For diagonal observables we have $a_{ij}=\alpha_i\delta_{ij}$, $a_{\ell m}'=\alpha_\ell'\delta_{\ell m}$ and $b_{kk'}=\beta_k\delta_{kk'}$.
Substituting these into \eqref{eq:explicit} gives
\[
\langle i|\mathcal{F}_{a_n,b_n}^{(n)}(a_{n+1})|j\rangle
= \alpha_i\delta_{ij}
\Bigl(\sum_k Q^{(n)}_i(k)\beta_k\Bigr)
\Bigl(\sum_\ell\Pi_{n;i\ell}\alpha_\ell'\Bigr),
\]
which is precisely the stated diagonal operator.
\end{proof}

Corollary~\ref{cor:classical} confirms that the quantum construction faithfully extends the classical HMM: diagonal observables correspond to classical random variables, and their expectations reproduce the standard Chapman–Kolmogorov equations.
Thus the isometric dilation provides a consistent \emph{quantum lifting} of the classical stochastic process, while Theorem~\ref{thm:main} reveals that this lifting, despite introducing Hilbert‑space structure, still retains enough classicality to make the causal order of emission and transition operationally indistinguishable for single‑step measurements.
This insight delineates a boundary between classical and genuinely quantum sequential models: to exploit quantum advantages such as superposition or entanglement in temporal processing, one must move beyond mere quantum encodings of classical dynamics and employ operations that violate the copying property responsible for the equality $\mathcal{F}=\mathcal{G}$.

\section{Discussion and Outlook}\label{Sec-disc}

Hidden quantum Markov models admit two genuinely distinct causal architectures once we decide how to time-order hidden dynamics and emissions. The same families of completely positive maps $\{\mathcal{E}_{H;n}\}_{n\ge 0}$ and $\{\mathcal{E}_{H,O;n}\}_{n\ge 0}$ can be wired as ``emission–then–transition'' (conventional) or ``transition–then–emission'' (causal). Because these expectations need not commute, the two architectures can generate different observable processes. Theorem~\ref{thm:non-quasi-equivalence} establishes a sharp separation: for a minimal qubit model with non-commuting hidden unitary $U$ ($0<|\theta|<\pi$) and sharp measurement, the observable states of the conventional and causal HQMMs are not quasi-equivalent. This holds for arbitrary initial hidden states $\phi_{H,0}^{\mathrm{conv}}$ and $\phi_{H,0}^{\mathrm{caus}}$ (not necessarily equal). Consequently, no finite-time measurement strategy can perfectly emulate the future predictions of the other model, regardless of how long one waits or how the hidden system is initialized. This asymptotic separation is reflected already at the level of one-step maps: the Choi–Jamio{\l}kowski operators of $\mathcal{F}^{(n)}_{a,b}$ and $\mathcal{G}^{(n)}_{a,b}$ are distinct for generic effects $a,b$, their diamond distance is strictly positive, and their entanglement spectra differ. Hence the causal ordering leaves a tangible imprint on the entanglement structure of the induced channels, which propagates to the level of infinite-time observable states.

At the same time, isometric lifts of classical hidden Markov models identify a sharp ``classical boundary''. Starting from a classical HMM $\lambda = (\boldsymbol{\pi}, \mathbf{\Pi}, \mathbf{Q})$, transitions and emissions are encoded into partial isometries $V_{H;n}$ and $V_{H,O;n}$ that preserve diagonality in a preferred basis \cite{ASS20,SS23}. In this regime, the emission–then–transition and transition–then–emission block maps $\mathcal{F}^{(n)}$ and $\mathcal{G}^{(n)}$ coincide on all diagonal observables. Hence the causal order becomes invisible at the level of classical statistics, even though off-diagonal operators may still encode distinct quantum temporal structures. This contrasts sharply with the non-commutative regime of Theorem~\ref{thm:non-quasi-equivalence}, where the two architectures remain asymptotically distinguishable.

These observations point to several concrete directions. Structurally, one may characterize which families $\{\mathcal{E}_{H;n}, \mathcal{E}_{H,O;n}\}$ force $\mathcal{F}^{(n)} \equiv \mathcal{G}^{(n)}$ beyond the diagonal, connecting to existing classifications of quantum Markovianity and process-tensor-based causal models \cite{MilzStrasberg21,WechsEtAl21b}. Operationally, conventional and causal architectures provide two natural hypotheses in quantum channel discrimination, where the performance gap directly quantifies the informational value of causal ordering \cite{PLLP19}. From a quantum memory perspective, our setup offers explicit testbeds for investigating hidden memory and memory compression in quantum implementations of stochastic processes \cite{TaElM,Elliot21}. Finally, causal HQMMs provide a natural language for representing matrix product states and valence-bond-type ground states, translating symmetry and locality constraints into finite-memory temporal processes \cite{AGLS24Q,Sou25}. This hybrid structure is particularly attractive for near-term experiments, where NISQ devices and hybrid quantum–classical schemes are already used to simulate noisy quantum channels and explore learning under realistic constraints \cite{David24,De22}.

\section*{Declarations}

\subsection*{Conflict of Interest:}
The authors declare no conflict of interest.

\subsection*{Data Availability:}
No data were generated or analyzed during this study.



\end{document}